\documentclass[11pt,amssymb,amsfont,a4paper]{article}
\usepackage{latexsym,amssymb,amsmath,amsthm,color}
\usepackage{geometry}
\usepackage{graphics}
\usepackage{url}
\usepackage{graphicx}
\usepackage{color}
\usepackage[utf8]{inputenc}
\usepackage{xcolor}
\usepackage{authblk}
\usepackage{multirow}
\usepackage{tikz}
\usepackage{pstricks,pst-node}

\theoremstyle{definition}
\newtheorem{definition}{Definition}
\newtheorem{example}{Example}

\theoremstyle{plain}
\newtheorem{theorem}{Theorem}
\newtheorem{proposition}{Proposition}
\newtheorem{lemma}{Lemma}
\newtheorem{remark}{Remark}
\newtheorem{corollary}{Corollary}

\title{On the similarities between generalized rank and Hamming weights and their applications to network coding}
\author{Umberto Mart{\'i}nez-Pe\~{n}as\thanks{umberto@math.aau.dk}}
\affil{Department of Mathematical Sciences, Aalborg University, Denmark}

\begin{document}

\maketitle

\begin{abstract}
Rank weights and generalized rank weights have been proven to characterize error and erasure correction, and information leakage in linear network coding, in the same way as Hamming weights and generalized Hamming weights describe classical error and erasure correction, and information leakage in wire-tap channels of type II and code-based secret sharing. Although many similarities between both cases have been established and proven in the literature, many other known results in the Hamming case, such as bounds or characterizations of weight-preserving maps, have not been translated to the rank case yet, or in some cases have been proven after developing a different machinery. The aim of this paper is to further relate both weights and generalized weights, show that the results and proofs in both cases are usually essentially the same, and see the significance of these similarities in network coding. Some of the new results in the rank case also have new consequences in the Hamming case. \\


\end{abstract}

\section{Introduction}

Linear network coding has been intensively studied during the last decade \cite{ahlswede, cai-yeung, random, errors-network, koetter2003, rgrw, linearnetwork, on-metrics, silva-universal, yang-char, yang-weight}. Consider a network with several sources and several sinks, where each source transmits several packets through the network to multiple sinks. Following \cite{ahlswede, random, koetter2003, linearnetwork}, ``linear network coding'' is defined as the process by which, in each node of the network, linear combinations of the received packets are generated (possibly at random \cite{random}) and sent (see \cite[Definition 1]{koetter2003}). We assume no delays nor cycles. 

In this context, errors are considered as erroneous packets that appear on some links, and erasures are considered as the deficiency of the rank of the matrix (called transfer matrix \cite{koetter2003, rgrw, silva-universal}) that describes the received packets as combinations of the ones sent by a given source \cite{rgrw, silva-universal}. In secure network coding, an adversary (or several) may compromise the security of the network by doing the following, among other attacks: introducing $ t $ erroneous packets on $ t $ different links, modifying the transfer matrix and obtaining information from the sent packets by wiretapping several links \cite{rgrw, on-metrics, silva-universal}. 

In classical coding for error and erasure correction \cite{pless}, coding for wire-tap channels of type II \cite{luo, ozarow, wei} and code-based secret sharing \cite{secure-computation, kurihara-secret, shamir}, the original message is encoded into a vector $ \mathbf{c} = (c_1, c_2, \ldots, c_n) \in \mathbb{F}_q^n $, where $ \mathbb{F}_q $ is some finite field. Then, errors, erasures and information leakage happen component-wise. This means that some components of $ \mathbf{c} $ may be wrong (errors), some components may be erased (erasures), and a wiretapping adversary may obtain some components (information leakage). Using source coding on a network, as in \cite{rgrw, on-metrics}, all this is considered to happen on some linear combinations: errors are wrong combinations, erasures are losses of combinations, and information leakage is considered in the form of leaked combinations.

In the classical case, Hamming weights \cite{pless} and generalized Hamming weights \cite{wei} have been proven to describe error and erasure correction and information leakage on wire-tap channels of type II. On the other hand, in recent years there have been several attempts to find a suitable weight and generalized weight to study linear network coding \cite{errors-network, rgrw, oggier, on-metrics, yang-char, yang-weight}. Finally, rank weights and generalized rank weights, introduced in \cite{gabidulin} and \cite{rgrw, oggier}, respectively, have been proven to describe exactly the worst case error and erasure correction capability \cite{rgrw, on-metrics, silva-universal}, and worst case information leakage on networks \cite{rgrw, silva-universal}. 

Many similarities between Hamming weights and rank weights have been considered since the paper \cite{gabidulin}, and for generalized ones since \cite{rgrw, oggier}. However, many results on Hamming weights still have no counterpart in the rank case, or require proofs using a different machinery.

The aim of this paper is to give some alternative definitions of rank weights \cite{gabidulin} and generalized rank weights \cite{jerome, slides, rgrw, oggier}, and then show that most of the well-known results for Hamming weights, classical error and erasure correction and information leakage, can be directly translated to rank weights, network error and erasure correction and information leakage on networks, once the right definitions and tools are introduced.
%

After giving some preliminary tools from the literature in Section \ref{defs}, the new results in this paper are distributed as follows: In Section \ref{sec equivalent defs}, we gather alternative definitions of rank weights and generalized rank weights from the literature, and propose some new definitions, proving the equivalence between them. In contrast with \cite{jerome, slides, oggier}, we also treat relative weights \cite{rgrw}. In Section \ref{sec eq}, we study linear equivalences of codes, that is, vector space isomorphisms between codes that preserve rank weights (and generalized rank weights), which allow to say when two codes perform exactly equally in secure network coding. We establish new characterizations of these equivalences that also give a connection with information leakage. We treat for the first time the case of different lengths and obtain the minimum possible lengths of codes, up to these equivalences. In Section \ref{sec bounds}, we establish a way to derive bounds on generalized rank weights from bounds on generalized Hamming weights, and give a list of some of these bounds. In the rest of the section, we discuss what the Singleton bound in the rank case can be, establishing a new alternative version. In Section \ref{sec puncturing}, we introduce the concept of rank-punctured codes, which plays the same role as classical punctured codes, and which are a main tool for the study of rank weights, erasure correction and information leakage, since punctured codewords are conceptually the same as codewords with erasures. We use this to characterize MRD ranks of codes and introduce the concept of information spaces. Finally, in Section \ref{sec secure}, we revisit some of the results regarding error and erasure correction and information leakage on networks. We obtain new relations regarding information leakage and duality, estimate information leakage in terms of dimensions of spaces, and propose a slightly different decoder than that of \cite{rgrw, on-metrics}, proving also the characterization of the correction capability of arbitrary (in particular, $ \mathbb{F}_q $-linear) coding schemes, which has not been stated nor proven yet.

\section{Definitions and preliminaries} \label{defs}

Let $ q $ be a prime power and $ m $ and $ n $, two positive integers. $ \mathbb{F}_q $ denotes the finite field with $ q $ elements. All vectors are considered to be row vectors, and we use the notation $ A^T $ to denote the transpose of a matrix $ A $.

\subsection{Linear network coding model} \label{subsec linear network}

We will consider the network model with errors in \cite{rgrw, on-metrics}, where the original message $ \mathbf{x} \in \mathbb{F}_{q^m}^k $ (considered as $ k $ packets in $ \mathbb{F}_{q^m} $) is encoded by a given source into $ \mathbf{c} \in \mathbb{F}_{q^m}^n $, whose $ n $ components (seen as packets) are sent through a network with $ n $ outgoing links from that source node and where a given receiver obtains $ \mathbf{y} = \mathbf{c} A^T + \mathbf{e} $, for some transfer matrix $ A \in \mathbb{F}_q^{N \times n} $ and some error vector $ \mathbf{e} \in \mathbb{F}_{q^m}^N $.

As in \cite{rgrw, on-metrics}, when treating error and erasure correction, we will consider multicast networks with one source and several sinks, and no delays nor cycles. In the noiseless case, for treating just information leakage to an adversary, we may assume several sources as long as the packets sent by different sources have no correlations. This allows to treat packets from a different source as errors, which give no extra information to a wiretapping adversary by \cite[Proposition 5]{rgrw}.

The length of the vector $ \mathbf{c} $ is defined as $ n $, and corresponds to the number of outgoing links from the source in the network, while $ m $ corresponds to the packet size. Therefore, $ m $ and $ n $ do not play a symmetric role. 

Although it is usual in the literature to only consider the case $ n \leq m $, we consider all cases, and we argue as follows (see also \cite[Section I.A]{rgrw} for more details): on the one hand, in some Internet protocols, the size of each packet ($ m $) is bounded by some parameters of the protocol, whereas the number of outgoing links ($ n $) is not necessarily bounded. On the other hand, since many computations are carried out over the extension field $ \mathbb{F}_{q^m} $, requiring $ m \geq n $ may extremely increase the computational complexity of the encoding and decoding. 

\subsection{Codes and coding schemes} \label{subsec coding schemes}

A code in $ \mathbb{F}_{q^m}^n $ is just a subset $ C \subset \mathbb{F}_{q^m}^n $, whose length is defined as $ n $. We say that $ C $ is linear (respectively $ \mathbb{F}_q $-linear) if it is an $ \mathbb{F}_{q^m} $-linear subspace (respectively $ \mathbb{F}_q $-linear). The term arbitrary is used for all codes, including non-linear codes. 

\begin{definition}[\textbf{\cite[Definition 7]{rgrw}}]
A coding scheme (or binning scheme) with message set $ \mathcal{S} $ is a family of disjoint nonempty subsets of $ \mathbb{F}_{q^m}^n $, $ \mathcal{P}_\mathcal{S} = \{ C_\mathbf{x} \}_{\mathbf{x} \in \mathcal{S}} $, together with a probability distribution over each of these sets. 
\end{definition}

\begin{definition}
A coding scheme as in the previous definition is said to be linear if $ \mathcal{S} = \mathbb{F}_{q^m}^\ell $, where $ 0 < \ell \leq n $, and
$$ \alpha C_\mathbf{x} + \beta C_\mathbf{y} \subset C_{\alpha \mathbf{x} + \beta \mathbf{y}}, $$
for all $ \alpha, \beta \in \mathbb{F}_{q^m} $ and all $ \mathbf{x}, \mathbf{y} \in \mathbb{F}_{q^m}^\ell $. Similarly in the $ \mathbb{F}_q $-linear case (where $ \mathcal{S} = \mathbb{F}_q^\ell $, $ 0 < \ell \leq mn $).
\end{definition}

The encoding in the coding scheme is given in \cite[Definition 7]{rgrw} as follows: for each $ \mathbf{x} \in \mathcal{S} $, we choose at random (with the chosen distribution) an element $ \mathbf{c} \in C_\mathbf{x} $. With these definitions, the concept of coding scheme generalizes the concept of code, since a code is a coding scheme where $ \# C_\mathbf{x} = 1 $, for each $ \mathbf{x} \in \mathcal{S} $, and thus no probability distribution is required. In the same way, linear and $ \mathbb{F}_q $-linear coding schemes generalize linear and $ \mathbb{F}_q $-linear codes, respectively.

An equivalent way to describe linear (and $ \mathbb{F}_q $-linear) coding schemes is by nested linear code pairs, introduced in \cite[Section III.A]{zamir}. We use the description in \cite[Subsection 4.2]{secure-computation}.

\begin{definition}[\textbf{\cite{secure-computation, zamir}}] \label{definition NLCP}
A nested linear code pair is a pair of linear codes $ C_2 \varsubsetneq C_1 \subset \mathbb{F}_{q^m}^n $. Choose a linear space $ W $ such that $ C_1 = C_2 \oplus W $ (where $ \oplus $ represents the direct sum of vector spaces) and an isomorphism $ \psi : \mathbb{F}_{q^m}^\ell \longrightarrow W $, where $ \ell = \dim(C_1/C_2) $. Then we define the sets $ C_\mathbf{x} = \psi(\mathbf{x}) + C_2 $. They form a linear coding scheme called nested coset coding scheme \cite{rgrw}.

If we choose the probability distribution to be uniform, then the encoding can be done as follows: Take uniformly at random $ \mathbf{c}^\prime \in C_2 $ and define $ \mathbf{c} = \psi(\mathbf{x}) + \mathbf{c}^\prime $.
\end{definition}

A given code $ C \subset \mathbb{F}_{q^m}^n $, seen as a pair $ 0 \varsubsetneq C $ is suitable for error correction, but is not suitable for protection against information leakage. Ozarow and Wyner proposed in \cite{ozarow} using the pair $ C \varsubsetneq \mathbb{F}_{q^m}^n $ for protection against information leakage on noiseless channels. The idea of nested linear code pairs was introduced in \cite{zamir} to protect against both information leakage and noise.

Independently, the same idea was implicitly used by Shamir \cite{shamir} and Massey \cite[Section 3.1]{secure-computation} to construct secret sharing schemes, and general nested linear code pairs were first used for this purpose in \cite[Section 4.2]{secure-computation}, where it is claimed in an informal way that they include all possible linear coding schemes. We now state this in a formal way, omitting the proof, which is straightforward. The $ \mathbb{F}_q $-linear case is completely analogous.

\begin{proposition}
Given a linear coding scheme $ \mathcal{P}_\mathcal{S} = \{ C_\mathbf{x} \}_{\mathbf{x} \in \mathcal{S}} $, define $ C_1 = \bigcup_{\mathbf{x} \in \mathcal{S}} C_{\mathbf{x}} $ and $ C_2 = C_{\mathbf{0}} $ (recall that $ \mathcal{S} = \mathbb{F}_{q^m}^\ell $). Then, $ C_1 $ and $ C_2 $ are linear codes in $ \mathbb{F}_{q^m}^n $ and
\begin{enumerate}
\item
$ C_2 \varsubsetneq C_1 $.
\item
The relation given in $ C_1 $ by $ \mathbf{c} \thicksim \mathbf{d} $ if, and only if, there exists $ \mathbf{x} \in \mathbb{F}_{q^m}^\ell $ such that $ \mathbf{c}, \mathbf{d} \in C_{\mathbf{x}} $, is an equivalence relation that satisfies the following:
$$ \mathbf{c} \thicksim \mathbf{d} \quad \Longleftrightarrow \quad \mathbf{c} - \mathbf{d} \in C_2. $$
In particular, $ \mathcal{P}_\mathcal{S} = C_1 / C_2 $.
\item
The map $ \mathbb{F}_{q^m}^\ell \longrightarrow \mathcal{P}_\mathcal{S} = C_1/C_2 : \mathbf{x} \longmapsto C_{\mathbf{x}} $ is a vector space isomorphism.
\end{enumerate}
In particular, if we take a subspace $ W \subset C_1 $ such that $ C_1 = C_2 \oplus W $, then we can canonically define an isomorphism $ \psi: \mathbb{F}_{q^m}^\ell \longrightarrow W $ by $ C_{\mathbf{x}} \cap W = \{ \psi(\mathbf{x}) \} $. Of course, it satisfies that $ C_{\mathbf{x}} = \psi(\mathbf{x}) + C_2 $.
\end{proposition}

On the other hand, if $ d : \mathbb{F}_{q^m}^n \times \mathbb{F}_{q^m}^n \longrightarrow \mathbb{N} $ is the rank (respectively Hamming) distance \cite{gabidulin} (respectively \cite{pless}), we define the minimum rank (respectively Hamming) distance of the coding scheme $ \mathcal{P}_\mathcal{S} $ as
\begin{equation}
d(\mathcal{P}_\mathcal{S}) = \min \{ d(\mathbf{c}_1, \mathbf{c}_2) \mid \mathbf{c}_1 \in C_{\mathbf{x}_1}, \mathbf{c}_2 \in C_{\mathbf{x}_2}, \mathbf{x}_1 \neq  \mathbf{x}_2 \}.
\label{minimum distance}
\end{equation}
For arbitrary codes we obtain the usual definition of minimum distance. For arbitrary coding schemes, it is basically the minimum of the distances between the sets $ C_{\mathbf{x}} $, $ \mathbf{x} \in \mathcal{S} $. 

For a linear coding scheme $ \mathcal{P}_\mathcal{S} $ and the Hamming distance $ d $, $ d(\mathcal{P}_\mathcal{S}) $ coincides with the minimum coset distance introduced in \cite{duursma} or the first relative generalized Hamming weight \cite{luo}. For a linear coding scheme and the rank distance, it coincides with the first relative generalized rank weight \cite{rgrw}.

\subsection{Rank weights and rank supports}

Now we turn to rank weights. We first observe the following obvious fact from linear algebra.

\begin{lemma} \label{change basis}
Let $ \alpha_1, \alpha_2, \ldots, \alpha_m $ and $ \beta_1, \beta_2, \ldots, \beta_m $ be two bases of $ \mathbb{F}_{q^m} $ over $ \mathbb{F}_q $, and let $ \mathbf{c} \in \mathbb{F}_{q^m}^n $ be a vector. It can be written in a unique way as
$$ \mathbf{c} = \sum_{i=1}^m \mathbf{c}_i \alpha_i = \sum_{i=1}^m \mathbf{d}_i \beta_i, $$
where $ \mathbf{c}_i, \mathbf{d}_i \in \mathbb{F}_q^n $. Moreover,
$$ \langle \mathbf{c}_1, \mathbf{c}_2, \ldots, \mathbf{c}_m \rangle_{\mathbb{F}_q} = \langle \mathbf{d}_1, \mathbf{d}_2 \ldots, \mathbf{d}_m \rangle_{\mathbb{F}_q} \subset \mathbb{F}_q^n. $$
\end{lemma}

\begin{definition}[\textbf{\cite{gabidulin}, \cite[Section II.D]{rgrw}}]
Choose one of such bases $ \alpha_1, \alpha_2, \ldots, \alpha_m $, and a vector $ \mathbf{c} \in \mathbb{F}_{q^m}^n $. We define the rank support \cite{rgrw} of $ \mathbf{c} $ as
$$ G(\mathbf{c}) = \langle \mathbf{c}_1, \mathbf{c}_2, \ldots, \mathbf{c}_m \rangle_{\mathbb{F}_q}, $$
where $ \mathbf{c} = \sum_{i=1}^m \mathbf{c}_i \alpha_i $ and $ \mathbf{c}_i \in \mathbb{F}_q^n $. The rank weight of $ \mathbf{c} $ \cite{gabidulin} is then $ {\rm wt_R} (\mathbf{c}) = \dim (G(\mathbf{c})) $.
\end{definition}

From the previous lemma it follows that $ G(\mathbf{c}) $ (and $ {\rm wt_R} (\mathbf{c}) $) does not depend on the choice of the basis. However, from now on, we fix one such basis $ \alpha_1, \alpha_2, \ldots, \alpha_m $.

\begin{definition}[\textbf{\cite[Definition 1]{slides}}]
For each linear subspace $ D \subset \mathbb{F}_{q^m}^n $, we define its rank support as $ G(D) = \sum_{\mathbf{d} \in D} G(\mathbf{d}) $ and its rank weight as $ {\rm wt_R}(D) = \dim(G(D)) $.
\end{definition}

\begin{remark} \label{vector to matrix}
We can associate each vector $ \mathbf{c} \in \mathbb{F}_{q^m}^n $ with a matrix over $ \mathbb{F}_q $, which we denote as follows:
\begin{displaymath}
\mu(\mathbf{c}) = \left( \begin{array}{cccc}
c_{1,1} & c_{1,2} & \ldots & c_{1,n} \\
c_{2,1} & c_{2,2} & \ldots & c_{2,n} \\
\vdots & \vdots & \ddots & \vdots \\
c_{m,1} & c_{m,2} & \ldots & c_{m,n} \\
\end{array} \right),
\end{displaymath}
where $ \mathbf{c} = \sum_{i=1}^m \alpha_i \mathbf{c}_i $ and $ \mathbf{c}_i = (c_{i,1}, c_{i,2}, \ldots, c_{i,n}) \in \mathbb{F}_q^n $. Note that $ \alpha_i \mathbf{e}_j $, where $ \mathbf{e}_j $ is the canonical basis of $ \mathbb{F}_{q^m}^n $, for $ i=1,2,\ldots, m $ and $ j=1,2, \ldots, n $, is a basis of $ \mathbb{F}_{q^m}^n $ over $ \mathbb{F}_q $. It follows that $ \mu : \mathbb{F}_{q^m}^n \longrightarrow \mathbb{F}_q^{m \times n} $ is an $ \mathbb{F}_q $-linear vector space isomorphism. Moreover, the rank support of $ \mathbf{c} $ is the row space of $ \mu(\mathbf{c}) $, which we denote by $ {\rm row}(\mu(\mathbf{c})) $, and the rank weight of $ \mathbf{c} $ is the rank of $ \mu(\mathbf{c}) $, denoted by $ {\rm Rk}(\mu(\mathbf{c})) $.

The rank weight of a subspace $ D \subset \mathbb{F}_{q^m}^n $ is then the rank of the matrix obtained by appending all rows of all matrices corresponding to the vectors in $ D $. It can be shown \cite[Proposition 3 (4)]{slides} that we can take the vectors in a basis of $ D $.
\end{remark}

Note that $ G(\mathbf{c}) = G(\langle \mathbf{c} \rangle) $ and thus $ {\rm wt_R}(\mathbf{c}) = {\rm wt_R}( \langle \mathbf{c} \rangle ) $, for every $ \mathbf{c} \in \mathbb{F}_{q^m}^n $.

\subsection{Trace codes, subfield codes and Galois closures}

Now we gather some tools from the literature regarding trace and subfield codes, and Galois closures. More details can be found in \cite{galoisinvariance}, \cite[Section 3.8]{pless}, \cite[Section II]{stichtenoth} or \cite[Chapter 9]{stichtenothbook}:

\begin{definition}
For a vector $ \mathbf{x} = (x_1, x_2, \ldots, x_n) \in \mathbb{F}_{q^m}^n $ and any integer $ i \geq 0 $, we define $ \mathbf{x}^{q^i} = (x_1^{q^i}, x_2^{q^i}, \ldots, x_n^{q^i}) $. Then we define the trace map on vectors as follows
$$ {\rm Tr} : \mathbb{F}_{q^m}^n \longrightarrow \mathbb{F}_q^n: \mathbf{x} \longmapsto \sum_{i=0}^{m-1} \mathbf{x}^{q^i}.  $$
For a linear subspace $ D \subset \mathbb{F}_{q^m}^n $, we define its Galois closure \cite[Definition]{stichtenoth} as 
$$ D^* = \sum_{i=0}^{m-1} D^{q^i}, $$
its trace code as $ {\rm Tr}(D) = \{ {\rm Tr}(\mathbf{d}) \mid \mathbf{d} \in D \} $ and its subfield code as $ D \vert_{\mathbb{F}_q} = D \cap \mathbb{F}_q^n $. We say that $ D $ is Galois closed if $ D = D^* $. If $ D \subset \mathbb{F}_q^n $ and is $ \mathbb{F}_q $-linear, we define its extended code as $ D \otimes \mathbb{F}_{q^m} $, that is, the code generated over $ \mathbb{F}_{q^m} $ by the set $ D $, also denoted as $ \langle D \rangle_{\mathbb{F}_{q^m}} \subset \mathbb{F}_{q^m}^n $.
\end{definition}

Note that $ {\rm Tr} $ is $ \mathbb{F}_q $-linear and $ D^* $ is the smallest Galois closed linear code containing $ D $ \cite{stichtenoth}. Moreover, a linear subspace $ D \subset \mathbb{F}_{q^m}^n $ is Galois closed if, and only if $ D^q \subset D $, which is equivalent to $ D^q = D $. 

The following proposition easily follows from \cite[Lemma 1]{stichtenoth}. The equivalence between items 1, 2, 4 and 5 were also noticed in \cite{galoisinvariance, slides}.

\begin{proposition} [\textbf{\cite{stichtenoth}}] \label{galois}
For every linear code $ C \subset \mathbb{F}_{q^m}^n $ of dimension $ k $, the following are equivalent:
\begin{enumerate}
\item
$ C $ is Galois closed.
\item
$ C $ admits a basis of vectors in $ \mathbb{F}_q^n $.
\item
$ C $ has a basis consisting of vectors of rank weight $ 1 $.
\item
$ C = C \vert_{\mathbb{F}_q} \otimes \mathbb{F}_{q^m} $.
\item
$ C = {\rm Tr}(C) \otimes \mathbb{F}_{q^m} $.
\item
$ {\rm Tr}(C) = C \vert _{\mathbb{F}_q} $.
\item
$ \dim({\rm Tr}(C)) = k $.
\item
$ \dim(C \vert _{\mathbb{F}_q}) = k $.
\end{enumerate}
\end{proposition}

We give a final tool due to Delsarte \cite[Theorem 2]{delsarte}:

\begin{lemma}[\textbf{Delsarte \cite{delsarte}}] \label{delsarte}
For every linear code $ C \subset \mathbb{F}_{q^m}^n $, we have that
$$ (C \vert_{\mathbb{F}_q})^\perp = {\rm Tr}(C^\perp), \quad \textrm{and} \quad (C^\perp) \vert_{\mathbb{F}_q} = ({\rm Tr}(C))^\perp. $$
\end{lemma}

\section{Equivalent definitions of rank weights and generalized rank weights} \label{sec equivalent defs}

In this section we give new equivalent definitions of generalized rank weights \cite{rgrw, oggier}. In contrast with \cite{jerome, slides, oggier}, we also treat relative weights \cite{rgrw}. Both have been proven to characterize worst-case information leakage and error and erasure correction on networks \cite{rgrw, oggier}.

\subsection{The Hamming case}

We briefly recall the definitions of Hamming weights, generalized Hamming weights \cite{wei} and their relative versions \cite{luo}. Following \cite[Section II]{wei} (see also \cite[Section 7.10]{pless}), given a linear subspace $ D \subset \mathbb{F}_{q^m}^n $, we define its support as $ {\rm Supp}(D) = \{ i \mid \exists \mathbf{d} \in D, d_i \neq 0 \} $ and its Hamming weight as $ {\rm wt_H}(D) = \# {\rm Supp}(D) $. The $ r $-th generalized Hamming weight of a code $ C $ \cite{wei}, and $ r $-th relative generalized Hamming weight of a nested linear code pair $ C_2 \varsubsetneq C_1 $ \cite{luo} are, respectively,
\begin{equation}
d_{H,r}(C) = \min \{ {\rm wt_H} (D) \mid D \subset C, \dim(D) = r \},  
\label{hamming}
\end{equation}
\begin{equation}
\begin{split}
M_{H,r}(C_1,C_2) = \min \{ & {\rm wt_H} (D) \mid D \subset C_1, \\
 & D \cap C_2 = 0, \dim(D) = r \}.
\label{hamming relative}
\end{split}
\end{equation}

\subsection{Existing equivalent definitions}

We briefly review the existing equivalent definitions of generalized rank weights and their relative versions. We attribute the following lemma to a combination of \cite{stichtenoth} with \cite{rgrw} for $ \dim(D) = 1 $, and a combination of \cite{stichtenoth} with \cite{slides} for the general case, and show why:

\begin{lemma} [\textbf{\cite{slides, rgrw, stichtenoth}}] \label{charact stic}
For any linear subspace $ D \subset \mathbb{F}_{q^m}^n $, 
$$ {\rm wt_R} (D) = {\rm wt_R} (D^*) = \dim({\rm Tr}(D)) = \dim (D^*). $$
\end{lemma}
\begin{proof}
It is immediate that $ \dim(D^*) = \dim({\rm Tr}(D^*)) $ from Proposition \ref{galois}, and moreover it holds that $ {\rm Tr}(D^*) = {\rm Tr}(D) $.

The equality $ {\rm wt_R} (D) = \dim (D^*) $ is proven in \cite[Lemma 11]{rgrw} for $ \dim(D) = 1 $, hence the result follows immediately in that case.

On the other hand, \cite[Theorem 16]{slides} states that $ G(D) = {\rm Tr}(D) $, hence $ {\rm wt_R} (D) = \dim({\rm Tr}(D)) $ and the result follows in the general case.
\end{proof}

Now we define generalized rank weights, introduced in \cite{oggier} for $ n \leq m $, and their relative versions, both introduced in general in \cite{rgrw}:

\begin{definition}[\textbf{\cite[Definition 2]{rgrw}}]
For a linear code $ C \subset \mathbb{F}_{q^m}^n $ and $ 1 \leq r \leq k = \dim(C) $, we define its $ r $-th generalized rank weight as
\begin{equation}
\begin{split}
d_{R,r}(C) = \min \{ & \dim V \mid V \subset \mathbb{F}_{q^m}^n, V = V^*, \\
 & \dim(C \cap V) \geq r \}.
\end{split}
\label{def1}
\end{equation}
For a nested linear code pair $ C_2 \varsubsetneq C_1 \subset \mathbb{F}_{q^m}^n $, we define its $ r $-th relative generalized rank weight as
\begin{equation}  
\begin{split}
M_{R,r}(C_1, C_2) = \min \{ & \dim V \mid V \subset \mathbb{F}_{q^m}^n, V = V^*, \\
 & \dim((C_1 \cap V)/(C_2 \cap V)) \geq r \}.
\label{def1r}
\end{split}
\end{equation}
\end{definition}

Fix a linear code $ C \subset \mathbb{F}_{q^m}^n $ and $ 1 \leq r \leq k = \dim(C) $. We have the following equivalent definitions from the literature:

\begin{lemma} [\textbf{\cite[Corollary 17]{slides}}]
The $ r $-th generalized rank weight $ d_{R,r}(C) $ is equal to
\begin{equation}
\min \{ {\rm wt_R} (D) \mid D \subset C, \dim(D) = r \}.
\label{def3}
\end{equation}
\end{lemma}

\begin{lemma} [\textbf{\cite[Proposition II.1]{jerome}}]
If $ n \leq m $, the $ r $-th generalized rank weight $ d_{R,r}(C) $ is equal to
\begin{equation}
\min \{ \max \{ {\rm wt_R}(\mathbf{x}) \mid \mathbf{x} \in D^* \} \mid D \subset C, \dim(D) = r \}.
\end{equation} \label{defoggier}
\end{lemma}

\subsection{New equivalent definitions}

In this subsection, we give new equivalent definitions of rank weights, generalized rank weights and their relative versions.

\begin{theorem} \label{weights basis}
For any linear subspace $ D \subset \mathbb{F}_{q^m}^n $, we have that
$$ {\rm wt_R} (D) = \min \{ {\rm wt_H}(\varphi_B(D)) \mid B \subset \mathbb{F}_q^n \textrm{ is a basis of } \mathbb{F}_{q^m}^n \}, $$
where $ \varphi_B : \mathbb{F}_{q^m}^n \longrightarrow \mathbb{F}_{q^m}^n $ is the linear map defined as $ \varphi_B(\mathbf{c}) = \mathbf{x} $, where $ \mathbf{c} = \sum_{i=1}^n x_i \mathbf{v}_i $ and $ B = \{ \mathbf{v}_i \}_{i=1}^n $. In particular, for every vector $ \mathbf{c} \in \mathbb{F}_{q^m}^n $, we have that
\begin{equation*}
\begin{split}
{\rm wt_R}(\mathbf{c}) = \min \{ & {\rm wt_H}(\mathbf{x}) \mid \mathbf{c} = \sum_{i=1}^n x_i \mathbf{v}_i, \\
 & B = \{\mathbf{v}_i \}_{i=1}^n \subset \mathbb{F}_q^n \textrm{ is a basis of } \mathbb{F}_{q^m}^n \}.
\end{split}
\end{equation*}
\end{theorem}

The following inequality is obtained when choosing the basis $ B $ as the canonical basis. It also follows easily from the definitions and was first noticed by Gabidulin \cite{gabidulin} when $ \dim(D) = 1 $:
\begin{equation}\label{ineq}
{\rm wt_R}(D) \leq {\rm wt_H}(D).
\end{equation} 

\begin{proof}[Proof of Theorem \ref{weights basis}]
We first prove the inequality $ \leq $: Let $ B = \{ \mathbf{v}_i \}_{i=1}^n \subset \mathbb{F}_q^n $ be a basis of $ \mathbb{F}_{q^m}^n $. If $ \mathbf{c} = \sum_{i=1}^n x_i \mathbf{v}_i $ and $ j \geq 0 $, then 
$$ \mathbf{c}^{q^j} = \left( \sum_{i=1}^n x_i \mathbf{v}_i \right)^{q^j} = \sum_{i=1}^n x_i^{q^j} \mathbf{v}_i^{q^j} = \sum_{i=1}^n x_i^{q^j} \mathbf{v}_i, $$
since $ \mathbf{v}_i \in \mathbb{F}_q^n $. It follows that $ \varphi_B(\mathbf{c}^{q^j}) = \varphi_B(\mathbf{c})^{q^j} $, for all $ \mathbf{c} \in \mathbb{F}_{q^m}^n $ and all $ j \geq 0 $, and therefore, 
$$ \varphi_B(D^*) = \sum_{j=0}^{m-1} \varphi_B(D^{q^j}) = \sum_{j=0}^{m-1} \varphi_B(D)^{q^j} = \varphi_B(D)^*. $$ 
Hence, using this and Lemma \ref{charact stic}, we see that
$$ {\rm wt_R}(D) = \dim(D^*) = \dim(\varphi_B(D^*)) $$
$$ = \dim(\varphi_B(D)^*) = {\rm wt_R}(\varphi_B(D)) \leq {\rm wt_H}(\varphi_B(D)), $$ 
where the last inequality follows from (\ref{ineq}).

Now we prove the inequality $ \geq $: We will show that we may select an appropriate basis $ B $ from the given family such that $ {\rm wt_R}(D) \geq {\rm wt_H}(\varphi_B(D)) $. 

By Proposition \ref{galois}, since $ D^* $ is Galois closed, it has a basis $ \mathbf{v}_1, \mathbf{v}_2, \ldots, \mathbf{v}_s $ of vectors in $ \mathbb{F}_q^n $. We may extend it to a basis $ B = \{ \mathbf{v}_1, \mathbf{v}_2, \ldots, \mathbf{v}_n \} $ of $ \mathbb{F}_q^n $, which is then a basis of $ \mathbb{F}_{q^m}^n $ as an $ \mathbb{F}_{q^m} $-linear space. Then $ {\rm Supp} (\varphi_B(D)) \subset \{ 1,2, \ldots s \} $, since $ \varphi_B (\mathbf{v}_i) = \mathbf{e}_i $, where the vectors $ \mathbf{e}_i $ constitute the canonical basis. Therefore, $ {\rm wt_R}(D) = \dim(D^*) = s \geq {\rm wt_H}(\varphi_B(D)) $, as desired, and the inequality follows.
\end{proof}

We now give the following new equivalent definitions of generalized rank weights:

\begin{theorem} \label{definitions}
For a linear code $ C \subset \mathbb{F}_{q^m}^n $ and $ 1 \leq r \leq k = \dim(C) $, the $ r $-th generalized rank weight of $ C $ is equal to:
\begin{equation}
\min \{ d_{H,r}(\varphi_B(C)) \mid B \subset \mathbb{F}_q^n \textrm{ is a basis of } \mathbb{F}_{q^m}^n \},
\label{def4}
\end{equation}
\begin{equation}
n - \max \{ \dim(L^G_U) \mid U \subset \mathbb{F}_{q^m}^k, \dim(U)=k-r \},
\label{def5}
\end{equation}
where $ G $ is a generator matrix of $ C $, $ \varphi_B $ is as in Theorem \ref{weights basis} and $ L^G_U = \{ \mathbf{x} \in \mathbb{F}_q^n \mid G \mathbf{x}^T \in U \} $.
\end{theorem}

Definition (\ref{def5}) is an analogous description as that of \cite[Lemma 1]{kloeve} for generalized Hamming weights, and is expressed in terms of a generator matrix of the code. We now give new equivalent definitions of relative generalized rank weights. Observe that Definition (\ref{def3r}) is an extension of Definition (\ref{def3}) for relative weights.

\begin{theorem} \label{definitions relative}
For a nested linear code pair $ C_2 \varsubsetneq C_1 \subset \mathbb{F}_{q^m}^n $ and $ 1 \leq r \leq \ell = \dim(C_1/C_2) $, the $ r $-th relative generalized rank weight of $ C_2 \varsubsetneq C_1 $ is equal to:
\begin{equation}
\min \{ {\rm wt_R} (D) \mid D \subset C_1, D \cap C_2 = 0, \dim(D) = r \},
\label{def3r}
\end{equation}
\begin{equation}
\min \{ M_{H,r}(\varphi_B(C_1), \varphi_B(C_2)) \mid B \subset \mathbb{F}_q^n \textrm{ is a basis of } \mathbb{F}_{q^m}^n \},
\label{def4r}
\end{equation}
\begin{equation}
\begin{split}
n - \max \{ & \dim(L^G_U) \mid U \subset \mathbb{F}_{q^m}^{k_1}, \dim(U)=k_1-r, \\
 & \dim(U^I) = k_2 \},
\end{split}
\label{def5r}
\end{equation}
where $ \varphi_B $ is as in Theorem \ref{weights basis}, $ G $ is a generator matrix of $ C_1 $, the first $ k_2 $ rows of $ G $ are a basis of $ C_2 $ and $ U^I $ is the projection of $ U $ onto the first $ k_2 $ coordinates.
\end{theorem}

Now, the last definition is analogous to \cite[Lemma 2]{zhuang-luo} for the Haming case. We only prove Theorem \ref{definitions relative}, since Theorem \ref{definitions} is obtained from it by choosing $ C_2 = 0 $.

\begin{proof}[Proof of Theorem \ref{definitions relative}]
We first prove $ (\ref{def1r}) \geq (\ref{def3r}) $: Take a $ V $ as in (\ref{def1r}). Since $ \dim((C_1 \cap V)/(C_2 \cap V)) \geq r $, we may choose a linear subspace $ D \subset C_1 \cap V $ such that $ \dim(D) = r $ and $ D \cap (C_2 \cap V) = 0 $. Hence $ D $ is as in (\ref{def3r}). Moreover, since $ D \subset V $, we have that $ D^* \subset V^* = V $, hence $ {\rm wt_R}(D) \leq \dim(V) $ by Lemma \ref{charact stic}, and the inequality follows.

No we prove $ (\ref{def1r}) \leq (\ref{def3r}) $: Take $ D $ as in (\ref{def3r}), and define $ V = D^* $, which is Galois closed. The natural linear map $ D \longrightarrow (C_1 \cap V)/(C_2 \cap V) $ is one to one, and hence $ \dim((C_1 \cap V)/(C_2 \cap V)) \geq \dim(D) = r $, and $ V $ is as in (\ref{def1r}). Moreover, $ \dim(V) = \dim(D^*) = {\rm wt_R}(D) $ by Lemma \ref{charact stic}, hence the inequality follows.

Using Theorem \ref{weights basis} and the expression (\ref{hamming relative}), we see that $ (\ref{def3r}) = (\ref{def4r}) $.

Finally, we prove that $ (\ref{def3r}) = (\ref{def5r}) $. Fix $ U \subset \mathbb{F}_{q^m}^{k_1} $ as in (\ref{def5r}), and define $ V = U^\perp $ and $ D = \{ \mathbf{v} G \mid \mathbf{v} \in V \} $. It holds that $ \dim(D) = r $ and $ D \cap C_2 = 0 $ since $ U^I = \mathbb{F}_{q^m}^{k_2} $. For any $ \mathbf{x} \in \mathbb{F}_q^n $, we have that
$$ G \mathbf{x}^T \in U \Longleftrightarrow \mathbf{v} G \mathbf{x}^T = \mathbf{0}, \forall \mathbf{v} \in V $$
$$ \Longleftrightarrow \mathbf{d} \cdot \mathbf{x} = \mathbf{0}, \forall \mathbf{d} \in D \Longleftrightarrow \mathbf{x} \in D^\perp, $$
and thus $ L^G_U = (D^\perp)\vert_{\mathbb{F}_q} $. Using Lemma \ref{charact stic} and Delsarte's Lemma \ref{delsarte},
$$ {\rm wt_R}(D) = \dim({\rm Tr}(D)) = n - \dim(L^G_U), $$
and we are done.
\end{proof}

\section{Equivalences of codes} \label{sec eq}

The purpose of this section is to characterize the $ \mathbb{F}_{q^m} $-linear vector space isomorphisms $ \phi : V \longrightarrow V^\prime $ that preserve rank weights, where $ V, V^\prime $ are Galois closed. 

Observe first of all that $ {\rm wt_R}(V) = \dim(V) $ and $ {\rm wt_R}(V^\prime) = \dim(V^\prime) $ by Lemma \ref{charact stic}, hence $ \dim(V) = \dim(V^\prime) $ is necessary if we want to preserve all possible rank weights.

A first characterization has been given in \cite[Theorem 1]{berger}, for $ V = V^\prime = \mathbb{F}_{q^m}^n $. We will see that, due to our new characterizations, equivalent codes are guaranteed to exactly perform in the same way in secure network coding, and not only regarding worst cases (which would be guaranteed just by having the same minimum rank distance). Moreover, in contrast with \cite{berger}, we consider equivalent codes with different lengths, which allows to consider equivalent codes that can be applied to networks with different number of outgoing links. As a consequence, we will see which is the minimum possible length of a code equivalent to a given one, that is, which is the minimum number of outgoing links that a given code requires.

\subsection{New characterizations}

Define the sets $ \Upsilon (\mathbb{F}_{q^m}^n) $ and $ \Lambda(\mathbb{F}_{q^m}^n) $ as the set of Galois closed linear subspaces of $ \mathbb{F}_{q^m}^n $ and the set of subspaces of the form $ V_I = \{ \mathbf{c} \in \mathbb{F}_{q^m}^n \mid c_i = 0, \forall i \notin I \} $, for some $ I \subset \mathcal{J} = \{1,2, \ldots, n \} $, respectively, as in \cite{rgrw}. We will write just $ \Upsilon $ and $ \Lambda $ if there is no confusion on the space $ \mathbb{F}_{q^m}^n $. For convenience, we also define $ L_I = \{ \mathbf{c} \in \mathbb{F}_{q}^n \mid c_i = 0, \textrm{ if } i \notin I \} $.

The rank weights are defined in terms of the spaces in $ \Upsilon $ (see (\ref{def1}) or \cite{rgrw}), and the Hamming weights are defined in terms of the spaces in $ \Lambda $ (see \cite{kurihara-secret, rgrw}). We will use this analogy in the rest of the paper.

We have the following two collections of characterizations of Hamming-weight and rank-weight preserving vector space isomorphisms. To the best of our knowledge, only the equivalence between items 2 and 5 has been noticed in the Hamming case. We only prove the rank case, that is, Theorem \ref{equiv rank}, being the proof of Theorem \ref{equiv hamming} analogous.

\begin{theorem} \label{equiv hamming}
Given an $ \mathbb{F}_{q^m} $-linear vector space isomorphism $ \phi : V \longrightarrow V^\prime $, where $ V \in \Lambda(\mathbb{F}_{q^m}^n) $ and $ V^\prime \in \Lambda(\mathbb{F}_{q^m}^{n^\prime}) $, the following are equivalent:
\begin{enumerate}
\item
If $ \mathbf{c} \in V $ and $ {\rm wt_H}(\mathbf{c}) = 1 $, then $ {\rm wt_H}(\phi(\mathbf{c})) = 1 $.
\item
$ \phi $ preserves Hamming weights, that is, $ {\rm wt_H}(\phi (\mathbf{c})) = {\rm wt_H}(\mathbf{c}) $, for all $ \mathbf{c} \in V $.
\item
For all linear subspaces $ D \subset V $, it holds that $ {\rm wt_H}(\phi (D)) = {\rm wt_H}(D) $.
\item
For all $ U \in \Lambda(\mathbb{F}_{q^m}^n) $, $ U \subset V $, it holds that $ \phi (U) \in \Lambda(\mathbb{F}_{q^m}^{n^\prime}) $.
\item
$ \phi $ is a monomial map. That is, if $ V = V_I $ and $ V^\prime = V_J $, with $ N = \# I = \# J $, then there exists a bijection $ \sigma : I \longrightarrow J $ and elements $ \gamma_1, \gamma_2, \ldots, \gamma_N \in \mathbb{F}_{q^m} $ such that $ \phi(\mathbf{e}_i) = \gamma_i \mathbf{e}_{\sigma (i)} $, for all $ i \in I $.
\end{enumerate}
In such case, we will say that $ \phi $ is a Hamming-weight preserving transformation or a Hamming equivalence. 
\end{theorem}

\begin{theorem} \label{equiv rank}
Given an $ \mathbb{F}_{q^m} $-linear vector space isomorphism $ \phi : V \longrightarrow V^\prime $, where $ V \in \Upsilon(\mathbb{F}_{q^m}^n) $ and $ V^\prime \in \Upsilon(\mathbb{F}_{q^m}^{n^\prime}) $, the following are equivalent:
\begin{enumerate}
\item
If $ \mathbf{c} \in V $ and $ {\rm wt_R}(\mathbf{c}) = 1 $, then $ {\rm wt_R}(\phi(\mathbf{c})) = 1 $.
\item
$ \phi $ preserves rank weights, that is, $ {\rm wt_R}(\phi (\mathbf{c})) = {\rm wt_R}(\mathbf{c}) $, for all $ \mathbf{c} \in V $.
\item
For all linear subspaces $ D \subset V $, it holds that $ {\rm wt_R}(\phi (D)) = {\rm wt_R}(D) $.
\item
For all $ U \in \Upsilon(\mathbb{F}_{q^m}^n) $, $ U \subset V $, it holds that $ \phi (U) \in \Upsilon(\mathbb{F}_{q^m}^{n^\prime}) $. 
\item
There exists $ \beta \in \mathbb{F}_{q^m}^* = \mathbb{F}_{q^m} \setminus \{ 0 \} $ and an $ \mathbb{F}_{q^m} $-linear vector space isomorphism $ \phi^\prime : V \longrightarrow V^\prime $ such that $ \phi^\prime(V \vert_{\mathbb{F}_q}) \subset V^\prime \vert_{\mathbb{F}_q} $ and $ \phi(\mathbf{c}) = \beta \phi^\prime (\mathbf{c}) $, for every $ \mathbf{c} \in V $. Equivalently, there exists a matrix $ A \in \mathbb{F}_q^{n \times n^\prime} $ and $ \beta \in \mathbb{F}_{q^m}^* $ such that $ \phi (\mathbf{c}) = \beta \mathbf{c} A $, for every $ \mathbf{c} \in V $.
\end{enumerate}
In such case, we will say that $ \phi $ is a rank-weight preserving transformation or a rank-metric equivalence. 
\end{theorem}
\begin{proof}
It is obvious that item 2 implies item 1 and item 3 implies item 2.

We now see that item 4 implies item 3. First, the number of sets in the family $ \Upsilon(\mathbb{F}_{q^m}^n) $ that are contained in $ V $ is the same as the number of sets in the family $ \Upsilon(\mathbb{F}_{q^m}^{n^\prime}) $ that are contained in $ V^\prime $, since $ \dim(V) = \dim(V^\prime) $. It follows that, given a linear subspace $ U \subset V $, $ U \in \Upsilon(\mathbb{F}_{q^m}^n) $ if, and only if, $ \phi (U) \in \Upsilon(\mathbb{F}_{q^m}^{n^\prime}) $. Now given a linear subspace $ D \subset V $, since $ D^* $ is the smallest set in $ \Upsilon(\mathbb{F}_{q^m}^n) $ that contains $ D $, it follows that $ \phi(D^*) = \phi(D)^* $. Therefore, $ {\rm wt_R}(D) = \dim(D^*) = \dim (\phi(D^*)) = \dim (\phi(D)^*) = {\rm wt_R}(\phi(D)) $ by Lemma \ref{charact stic}.

To prove that item 5 implies item 4, it is enough to show that, for a given subspace $ U \subset V $, if $ U^q \subset U $, then $ \phi(U)^q \subset \phi(U) $. Take bases $ B= \{ \mathbf{v}_1, \mathbf{v}_2, \ldots, \mathbf{v}_N \} $ and $ B^\prime = \{ \mathbf{v}^\prime_1, \mathbf{v}^\prime_2, \ldots, \mathbf{v}^\prime_N \} $ of $ V $ and $ V^\prime $ in $ \mathbb{F}_q^n $, respectively, such that $ \phi(\mathbf{v}_i) = \beta \mathbf{v}^\prime_i $. Take $ \mathbf{u} \in U $, and write it as $ \mathbf{u} = \sum_{i,j} \lambda_{i,j} \alpha_j \mathbf{v}_i $, where $ \lambda_{i,j} \in \mathbb{F}_q $. Then $ \phi(\mathbf{u})^q = \sum_{i,j} \lambda_{i,j} \beta^q \alpha_j^q \mathbf{v}^\prime_i $. Since $ \phi(\mathbf{u}^q) = \sum_{i,j} \lambda_{i,j} \beta \alpha_j^q \mathbf{v}^\prime_i \in \phi(U) $, it follows that $ \phi(\mathbf{u})^q \in \phi(U) $.

Finally, we prove that item 1 implies item 5, which is a slight modification of the proof given in \cite{berger}. Taking a basis of $ V $ in $ \mathbb{F}_q^n $ as before, it holds that $ \phi(\mathbf{v}_i) = \beta_i \mathbf{u}_i $, for some $ \mathbf{u}_i \in \mathbb{F}_q^n $ and $ \beta_i \in \mathbb{F}_{q^m}^* $. Since $ \phi $ is an isomorphism, the vectors $ \mathbf{u}_i $ are linearly independent.

Now take $ i \neq j $ and assume that $ \beta_i \neq a_{i,j} \beta_j $, for every $ a_{i,j} \in \mathbb{F}_q $. Then there exists a basis of $ \mathbb{F}_{q^m} $ over $ \mathbb{F}_q $ that contains $ \beta_i $ and $ \beta_j $. Therefore $ \phi(\mathbf{v}_i + \mathbf{v}_j) = \beta_i \mathbf{u}_i + \beta_j \mathbf{u}_j $, but $ {\rm wt_R}(\phi(\mathbf{v}_i + \mathbf{v}_j)) = {\rm wt_R}(\mathbf{v}_i + \mathbf{v}_j) = 1 $ and also $ {\rm wt_R}(\beta_i \mathbf{u}_i + \beta_j \mathbf{u}_j) = 2 $, since $ \mathbf{u}_i $ and $ \mathbf{u}_j $ are linearly independent. 

We have reached an absurd, so there exists $ a_{i,j} \in \mathbb{F}_q^* $ such that $ \beta_i = a_{i,j} \beta_j $, for all $ i,j $. Defining $ \beta = \beta_1 = a_{1,j} \beta_j $ and $ \mathbf{v}^\prime_i = a_{1,i}^{-1}\mathbf{u}_i $, we obtain a description of $ \phi $ as in item 5.
\end{proof}

This motivates the following definition.

\begin{definition}
We say that two (arbitrary) codes $ C \subset \mathbb{F}_{q^m}^n $ and $ C^\prime \subset \mathbb{F}_{q^m}^{n^\prime} $ are rank-metric equivalent if there exists a rank-metric equivalence $ \phi $ between $ V $ and $ V^\prime $ such that $ \phi (C) = C^\prime $, where $ C \subset V \in \Upsilon (\mathbb{F}_{q^m}^n) $ and $ C^\prime \subset V^\prime \in \Upsilon (\mathbb{F}_{q^m}^{n^\prime}) $. Similarly for Hamming equivalent codes.
\end{definition}

\begin{remark}
Observe that item 2 states that equivalent codes behave exactly in the same way regarding error and erasure correction, and not just in worst cases, since corresponding codewords have the same rank weight (see \cite[Subsection IV.C]{on-metrics} for MRD codes, and \cite[Theorem 4]{rgrw} and \cite[Theorem 2]{silva-universal} in general). On the other hand, item 4 states that equivalent linear codes behave exactly in the same way regarding information leakage, and not only in worst cases, since the information leaked by wiretapping links is measured by the dimension of $ C \cap U $, for some $ U \in \Upsilon $, as stated in \cite[Lemma 7]{rgrw}. The previous theorem thus states that one property is preserved if, and only if, the other is preserved.

The same holds for the Hamming case, where item 4 states that equivalent codes behave exactly in the same way regarding information leakage in code-based secret sharing \cite{one-point, kurihara-secret}, and item 2 states that equivalent codes behave exactly in the same way regarding usual error and erasure correction.

Item 1 states that it is only necessary for codes to be equivalent that they behave in the same way regarding ``unitary'' errors.
\end{remark}

\begin{remark}
Observe that, due to the equivalence between items 2 and 3, rank weight preserving transformations preserve not only minimum rank distances and rank weight distributions, but also generalized rank weights and generalized rank weight distributions.
\end{remark}

\begin{remark}
In the Hamming case, if $ \phi : C_1 \longrightarrow C_2 $ is an $ \mathbb{F}_{q^m} $-linear vector space isomorphism that preserves Hamming weights, for arbitrary linear codes $ C_1 \subset \mathbb{F}_{q^m}^n $ and $ C_2 \subset \mathbb{F}_{q^m}^{n^\prime} $, then it can be extended to a Hamming weight preserving isomorphism $ \widetilde{\phi} : V_I \longrightarrow V_J $, where $ I = {\rm Supp}(C_1) $ and $ J = {\rm Supp}(C_2) $. This is known as MacWilliams extension theorem (see \cite[Section 7.9]{pless}).

However, this is not true in the rank case. For a counterexample, see \cite[Example 2.9 (c)]{barra}.
\end{remark}

As a consequence, we can now establish the following relations between Hamming and rank weights:

\begin{theorem} \label{defs equivs}
For any linear codes $ D, C \subset \mathbb{F}_{q^m}^n $, we have that
\begin{equation*}
\begin{split}
{\rm wt_R}(D) = \min \{ & {\rm wt_H}(\phi(D)) \mid \phi : \mathbb{F}_{q^m}^n \longrightarrow \mathbb{F}_{q^m}^n \\
 & \textrm{ is a rank-metric equivalence} \},
\end{split}
\end{equation*}
\begin{equation*}
\begin{split}
d_{R,r}(C) = \min \{ & d_{H,r}(\phi(C)) \mid \phi : \mathbb{F}_{q^m}^n \longrightarrow \mathbb{F}_{q^m}^n \\
 & \textrm{ is a rank-metric equivalence} \},
\end{split}
\end{equation*}
where $ 1 \leq r \leq k = \dim(C) $. Moreover, if $ n \leq m $, we have that
\begin{equation*}
\begin{split}
{\rm wt_H}(D) = \max \{ & {\rm wt_R}(\phi(D)) \mid \phi : \mathbb{F}_{q^m}^n \longrightarrow \mathbb{F}_{q^m}^n \\
 & \textrm{ is a Hamming equivalence} \},
\end{split}
\end{equation*}
\begin{equation*}
\begin{split}
d_{H,k}(C) = \max \{ & d_{R,k}(\phi(C)) \mid \phi : \mathbb{F}_{q^m}^n \longrightarrow \mathbb{F}_{q^m}^n \\
 & \textrm{ is a Hamming equivalence} \}.
\end{split}
\end{equation*}
\end{theorem}
\begin{proof}
The second equality follows from the first one, which we now prove. By Theorem \ref{equiv rank}, the map $ \varphi_B $ in Theorem \ref{weights basis} is a rank-metric equivalence, for any basis $ B \subset \mathbb{F}_q^n $ of $ \mathbb{F}_{q^m}^n $, since it maps vectors in $ \mathbb{F}_q^n $ to vectors in $ \mathbb{F}_q^n $. On the other hand, given a rank-metric equivalence $ \phi : \mathbb{F}_{q^m}^n \longrightarrow \mathbb{F}_{q^m}^n $, with $ \beta $ and $ \phi^\prime $ as in item 5 in Theorem \ref{equiv rank}, define $ \mathbf{v}_i = \phi^{\prime -1}(\mathbf{e}_i) $, where $ \mathbf{e}_i $ is the $ i $-th vector in the canonical basis and $ B = \{ \mathbf{v}_i \}_{i=1}^n $. Hence $ \phi^\prime = \varphi_B $ and $ \phi = \beta \varphi_B $. Multiplication by $ \beta $ preserves Hamming weights, and hence we see that the first equality follows from Theorem \ref{weights basis}.

The last equality follows from the third one, which we now prove. First, for every Hamming equivalence $ \phi $, it follows from Theorem \ref{equiv hamming} and Equation (\ref{ineq}) that $ {\rm wt_H}(D) = {\rm wt_H}(\phi (D)) \geq {\rm wt_R}(\phi(D)) $, and therefore the inequality $ \geq $ follows. 

To conclude, we need to prove that there exists a Hamming equivalence $ \phi $ such that $ {\rm wt_H}(D) = {\rm wt_R}(\phi(D)) $. By taking a suitable Hamming equivalence, we may assume that $ D $ has a generator matrix $ G $ of the following form: the rows in $ G $ (a basis for $ D $) are $ \mathbf{g}_1, \mathbf{g}_2, \ldots, \mathbf{g}_r $, and there exist $ 0 = t_0 < t_1 < t_2 < \ldots < t_r \leq n $ such that, for every $ i = 1,2, \ldots, r $, $ g_{i,j} = 1 $ if $ t_{i-1} < j \leq t_i $, and $ g_{i,j} = 0 $ if $ t_i < j $. Observe that $ t_r = {\rm wt_H}(D) $.

Finally, choose a basis $ \gamma_1, \gamma_2, \ldots, \gamma_m $ of $ \mathbb{F}_{q^m} $ over $ \mathbb{F}_q $, and define the Hamming equivalence $ \phi(c_1,c_2, \ldots, c_n) = (\gamma_1 c_1, \gamma_2 c_2, \ldots, \gamma_n c_n) $. Then, $ \phi(D) $ has a generator matrix whose rows are $ \mathbf{h}_i = \phi(\mathbf{g}_i) $, which satisfy that $ h_{i,j} = \gamma_j $ if $ t_{i-1} < j \leq t_i $, and $ h_{i,j} = 0 $ if $ t_i < j $. 

It follows that $ G(\phi(D)) = \sum_{i=1}^r G(\mathbf{h}_i) = V_I $, where $ I = \{ 1,2, \ldots, t_r \} $, and we are done.
\end{proof}

\subsection{Rank degenerateness and minimum length}

Now we turn to degenerate codes in the rank case, extending the study in \cite[Section 6]{slides}. 

\begin{definition}
A linear code $ C \subset \mathbb{F}_{q^m}^n $ is rank degenerate if it is rank-metric equivalent to a linear code $ C^\prime \subset \mathbb{F}_{q^m}^{n^\prime} $ with $ n^\prime < n $.
\end{definition}

Hamming degenerate codes are defined in the analogous way. As in the Hamming case, rank degenerate codes are identified by looking at their last generalized rank weight. This is the definition of rank degenerate codes used in \cite{slides}. However, note that our definition actually states whether a given code does not require the given length, which in network coding means whether a code can be implemented with less outgoing links from the source node. 

The next proposition actually gives the whole range of lengths of linear codes rank-metric equivalent to a given one. To prove it, for every $ V \in \Upsilon(\mathbb{F}_{q^m}^n) $ and every basis $ B \subset \mathbb{F}_q^n $ of $ V $, we define the $ \mathbb{F}_{q^m} $-linear map 
\begin{equation} \label{psi basis}
\psi_B: V \longrightarrow \mathbb{F}_{q^m}^{\dim(V)} 
\end{equation}
given by $ \psi_B(\mathbf{c}) = \mathbf{x} $, if $ B = \{ \mathbf{v}_i \}_{i=1}^{\dim(V)} $ and $ \mathbf{c} = \sum_{i=1}^{\dim(V)} x_i \mathbf{v}_i $. It is a rank-metric equivalence by Theorem \ref{equiv rank}.

\begin{proposition} \label{range for which equivalent}
Given a linear code $ C \subset \mathbb{F}_{q^m}^n $ of dimension $ k $ and any positive integer $ n^\prime $, there exists a linear code $ C^\prime \subset \mathbb{F}_{q^m}^{n^\prime} $ that is rank-metric equivalent to $ C $ if, and only if, $ n^\prime \geq d_{R,k}(C) $.
\end{proposition}
\begin{proof}
For a given $ n^\prime $, assume that there exists a linear code $ C^\prime \subset \mathbb{F}_{q^m}^{n^\prime} $ that is rank-metric equivalent to $ C $. Then $ C^\prime $ has dimension $ k $ and $ d_{R,k}(C) = d_{R,k}(C^\prime) \leq n^\prime $.

Now fix $ n^\prime = d_{R,k}(C) = \dim(C^*) $. Take $ V = C^* $ and $ \psi_B $ as in (\ref{psi basis}) for some basis $ B \subset \mathbb{F}_q^n $ of $ V $. As remarked before, $ \psi_B $ is a rank-metric equivalence and thus $ C $ is rank-metric equivalent to $ C^\prime = \psi_B(C) \subset \mathbb{F}_{q^m}^{n^\prime} $.

Finally, take $ n^{\prime \prime} \geq n^\prime = d_{R,k}(C) $ and $ C^\prime $ as in the previous paragraph. Append $ n^{\prime \prime} - n^\prime \geq 0 $ zeroes to every codeword in $ C^\prime $. The obtained code $ C^{\prime \prime} \subset \mathbb{F}_{q^m}^{n^{\prime \prime}} $ is linear and rank-metric equivalent to $ C^\prime $, and thus also to $ C $, and we are done.
\end{proof}

Therefore, $ d_{R,k}(C) $ gives the minimum possible length (minimum number of outgoing links required by $ C $) of a linear code that is rank equivalent to $ C $. As an immediate consequence, we obtain the following:

\begin{corollary}
A linear code $ C \subset \mathbb{F}_{q^m}^n $ of dimension $ k $ is rank degenerate if, and only if, $ d_{R,k}(C) < n $, or equivalently, $ C^* \neq \mathbb{F}_{q^m}^n $.
\end{corollary}

On the other hand, we obtain the following result. The first part is \cite[Corollary 30]{slides}.

\begin{proposition} \label{when degenerate}
If $ mk < n $, then every linear code $ C \subset \mathbb{F}_{q^m}^n $ of dimension $ k $ is rank degenerate. On the other hand, if $ mk \geq n $, then there exists a linear code $ C \subset \mathbb{F}_{q^m}^n $ of dimension $ k $ that is not rank degenerate.
\end{proposition}
\begin{proof}
The first part follows from the previous corollary and the fact that $ \dim(C^*) \leq mk $.

Now, if $ mk \geq n $, choose $ \lambda_{l,j}^{(i)} \in \mathbb{F}_q $, for $ 1 \leq i \leq k $, $ 1 \leq j \leq n $ and $ 1 \leq l \leq m $, such that $ \langle \{ \mathbf{x}_{l,i} \}_{1 \leq l \leq m}^{1 \leq i \leq k} \rangle = \mathbb{F}_q^n $, where $ \mathbf{x}_{l,i} = \sum_{j=1}^n \lambda_{l,j}^{(i)} \mathbf{e}_j $ and $ \mathbf{e}_j $ is the canonical basis of $ \mathbb{F}_q^n $. This is possible since $ mk \geq n $.

On the other hand, define $ \mathbf{u}_i = \sum_{l=1}^m \alpha_l \mathbf{x}_{l,i} \in \mathbb{F}_{q^m}^n $, and $ C^\prime = \langle \mathbf{u}_1, \mathbf{u}_2, \ldots, \mathbf{u}_k \rangle $. Then, $ C^{\prime *} = \mathbb{F}_{q^m}^n $ and $ \dim(C^\prime) \leq k $. Taking $ C^\prime \subset C $, with $ \dim(C) =k $, we obtain the desired code.
\end{proof}

\section{Bounds on generalized rank weights} \label{sec bounds}

In this section we establish a method to derive bounds on generalized rank weights from bounds on generalized Hamming weights, and afterwards we discuss what the Singleton bound can be for generalized rank weights. Due to \cite[Lemma 7 and Theorem 2]{rgrw}, bounds on generalized rank weights directly translate into bounds on worst case information leakage on networks, and therefore are of significant importance. 

\subsection{Translating bounds on GHWs to bounds on GRWs}

Some attempts to give bounds similar to the ones in the Hamming case have been made \cite{jerome, rgrw, oggier}. In this subsection, we prove that most of the bounds in the Hamming case can be directly translated to the rank case.

Note that, since rank weights are smaller than or equal to Hamming weights (by Equation (\ref{ineq})), every bound of the form
$$ M \geq g_{s_1,s_2,\ldots,s_N}(d_{s_1} (C), d_{s_2}(C), \ldots, d_{s_N}(C)), $$
that is valid for Hamming weights, where $ M > 0 $ is a fixed positive real number and $ g_{s_1,s_2,\ldots,s_N} $ is increasing in each component, is obviously also valid for rank weights. This is the case of the classical Singleton or Griesmer bounds \cite[Section 7.10]{pless}. On the other hand, the next result is not straightforward if we do not use (\ref{def4}) or (\ref{def4r}). 

\begin{theorem} \label{comparison bounds}
Fix numbers $ k $ and $ 1 \leq r, s \leq k $, and functions $ f_{r,s}, g_{r,s} : \mathbb{N} \longrightarrow \mathbb{R} $, which may also depend on $ n, m, k $ and $ q $. If $ g_{r,s} $ is increasing, then every bound of the form 
$$ f_{r,s}(d_r (C)) \geq g_{r,s}(d_s (C)) $$
that is valid for generalized Hamming weights, for any linear code $ C \subset \mathbb{F}_{q^m}^n $ with $ \dim(C) = k $, is also valid for generalized rank weights. The same holds for relative weights.
\end{theorem}
\begin{proof}
By Theorem \ref{definitions}, there exists a basis $ B \subset \mathbb{F}_q^n $ of $ \mathbb{F}_{q^m}^n $ such that $ d_{R,r} (C) = d_{H,r}(\varphi_B(C)) $. Therefore,
$$ f_{r,s}(d_{R,r}(C)) = f_{r,s}(d_{H,r}(\varphi_B(C))) $$
$$ \geq g_{r,s}(d_{H,s}(\varphi_B(C))) \geq g_{r,s}(d_{R,s}(C)), $$
where the last inequality follows again from Theorem \ref{definitions}. Similarly for relative weights.
\end{proof}

\begin{remark}
The previous theorem is also valid, with the same proof, for the more general bounds
$$ f_{r,s_1,s_2,\ldots,s_N}(d_r (C)) $$
$$ \geq g_{r,s_1,s_2,\ldots,s_N}(d_{s_1} (C), d_{s_2}(C), \ldots, d_{s_N}(C)), $$
where $ g_{r,s_1,s_2,\ldots,s_N} $ is increasing in each component. However, most of the bounds in the literature are of the form of the previous theorem.
\end{remark}

In \cite{kloeve} and \cite[Part I, Section III.A]{tsfasman}, many of these kind of bounds are given for generalized Hamming weights. One of these (a particular case of \cite[Corollary 3.6]{tsfasman}) is proven for rank weights in \cite[Proposition II.3]{jerome}, using (\ref{def1}). Some of these are also valid for relative weights (see \cite[Proposition 1 and Proposition 2]{zhuang-luo} or \cite{new-relative}). We next list some of these bounds, where $ 1 \leq r \leq s \leq k $, and $ d_j = d_{R,j}(C) $, for all $ j $. Note that monotonicity is one of these bounds, and therefore it does not need a specific proof. Also recall that linear codes in this paper are $ \mathbb{F}_{q^m} $-linear, and hence the field size is $ q^m $, not $ q $.
\begin{enumerate}
\item
Monotonicity: $$ d_{r+1} \geq d_{r} + 1, $$
\item 
Griesmer-type (\cite[bound (14)]{tsfasman}): $$ d_r \geq \sum_{i=0}^{r-1} \left\lceil \frac{d_1}{q^{mi}} \right\rceil, $$
\item
Griesmer-type (\cite[bound (16)]{tsfasman}): $$ d_s \geq d_r + \sum_{i=0}^{s-r} \left\lceil \frac{(q^m-1)d_r}{(q^{mr}-1)q^{mi}} \right\rceil, $$
\item 
\cite[Theorem 1]{kloeve} or \cite[bound (18)]{tsfasman}: $$ (q^{ms}-1)d_r \leq (q^{ms} - q^{m(s-r)})d_s, $$
\item 
\cite[Corollary 1]{kloeve}: $$ (q^{mr}-1)d_1 \leq (q^{mr} - q^{m(r-1)})d_r, $$
\item 
\cite[Proposition II.3]{jerome}: $$ (q^{mr}-1)d_{r-1} \leq (q^{mr} - q^{m})d_r, $$
\item 
\cite[bound (20)]{tsfasman}: $$ d_r \geq n - \left\lfloor \frac{(q^{m(k-r)}-1)(n-d_s)}{q^{m(k-s)}-1} \right\rfloor. $$
\end{enumerate}

\begin{remark}
A trivial lower bound that is valid for every linear code is $ d_{R,r}(C) \geq r $, for all $ 1 \leq r \leq k $. Observe that a linear code $ C $ satisfies that $ d_{R,r}(C) = r $, for every $ 1 \leq r \leq k $ if, and only if, $ C $ is Galois closed. This gives another characterization of Galois closed spaces to those in Proposition \ref{galois}, in terms of generalized rank weights. In the Hamming case, $ d_{H,r}(C) = r $, for every $ 1 \leq r \leq k $ if, and only if, $ C = V_I $, for some $ I \subset \{ 1,2, \ldots, n \} $.
\end{remark}

\subsection{On the Singleton bound}

In this subsection, we discuss the possible extensions of the Singleton bound to rank weights. We start by giving a brief overview of the bounds in the literature that resemble the usual Singleton bound, both for a linear code $ C \subset \mathbb{F}_{q^m}^n $ and a nested linear code pair $ C_2 \varsubsetneq C_1 \subset \mathbb{F}_{q^m}^n $:

\begin{displaymath}
d_{R,r}(C) \leq \left\lbrace
\begin{array}{l}
 n - k + r \textrm{ \cite{rgrw}}, \\
 (m-1)k + r \textrm{ \cite{rgrw}}, \\
 \frac{m}{n}(n-k) +1 \textrm{, if $ r=1 $ \cite{Loi}}, \\
\end{array}\right.
\end{displaymath}

\begin{displaymath}
M_{R,s}(C_1,C_2) \leq \left\lbrace \begin{array}{l}
 n - k_1 + s \textrm{ \cite{rgrw}}, \\
 (m-1)(k_1-k_2) + s \textrm{ \cite{rgrw}}, \\
 \frac{m(n-k_1)}{n-k_2} +1 \textrm{, if $ s=1 $ \cite{rgrw}}, \\
\end{array} \right.
\end{displaymath}
where $ 1 \leq r \leq k = \dim(C) $ and $ 1 \leq s \leq k_1 - k_2 $, $ k_1 = \dim(C_1) $ and $ k_2 = \dim(C_2) $.

In \cite[Proposition 6]{oggiercyclic}, a refinement of the classical Singleton bound is given for cyclic codes. By \cite[Proposition 5]{oggiercyclic} and duality \cite[Theorem]{jerome}, this bound is $ d_{R,1}(C) \leq d_{R,k}(C) - k + 1 $. Hence this bound is implied by the classical bound and Proposition \ref{range for which equivalent}, or by monotonicity. The description in \cite{oggiercyclic} gives then an alternative description of this bound for cyclic codes.

As a tool for future bounds, we establish the following one. It shows how to obtain bounds for all generalized weights from bounds on the first one or the last one.

\begin{lemma}
For every linear code $ C \subset \mathbb{F}_{q^m}^n $, and for every $ 1 \leq r \leq k - 1 $, $ k = \dim(C) $, it holds that
$$ 1 \leq d_{R,r+1}(C) - d_{R,r}(C) \leq m. $$
The same bound applies to relative generalized rank weights.
\end{lemma}
\begin{proof}
It is enough to prove that, if $ D \subset D^\prime $ and $ \dim(D^\prime) = \dim(D) + 1 $, then $ {\rm wt_R}(D^\prime) \leq {\rm wt_R}(D) + m $. Take $ \mathbf{d} \in D^\prime $ such that $ D^\prime = D \oplus \langle \mathbf{d} \rangle $. Then $ D^{\prime *} = D^* + \langle \mathbf{d} \rangle^* $, and the result follows, since $ {\rm wt_R}(\mathbf{d}) \leq m $.
\end{proof}

Note that this bound implies that an inverse statement to Theorem \ref{comparison bounds} is not possible: Take for instance $ m = 1 $, then we have the bound $ d_{R,r+1} = d_{R,r} + 1 $, which holds for all linear codes. However, the bound $ d_{H,r+1} = d_{H,r} + 1 $ does not hold for all linear codes. 

The case $ r=1 $ of the following bound was established and proven by Loidreau in \cite{Loi} and for relative weights by Kurihara et al. in \cite[Proposition 3]{rgrw}. The general case follows from these and the previous lemma.

\begin{proposition}[\textbf{Alternative Singleton bound}] \label{special singleton}
If $ n > m $, then for every linear code $ C \subset \mathbb{F}_{q^m}^n $, and every $ 1 \leq r \leq k=\dim(C) $,
$$ d_{R,r}(C) \leq \frac{m}{n}(n-k) + m(r-1) + 1. $$
For a code pair $ C_2 \varsubsetneq C_1 \subset \mathbb{F}_{q^m}^n $, with $ k_i = \dim(C_i) $, $ i=1,2 $, and every $ 1 \leq r \leq \dim(C_1/C_2) $,
$$ M_{R,r}(C_1,C_2) \leq \frac{m(n-k_1)}{n-k_2} + m(r-1) + 1. $$
\end{proposition}

Now, for generalized rank weights, it is easy to see that this bound is sharper than the usual Singleton bound if, and only if,
\begin{equation}
r \leq \left\lfloor \frac{n(n-1) - (n-m)k}{n(m-1)} \right\rfloor ,
\label{r bound}
\end{equation}
which is a number in $ (1,k] $ if $ n \leq mk $ (the case where the code is not necessarily rank degenerate, see Proposition \ref{when degenerate}). However, as it is usual and for convenience, we give the following definition:

\begin{definition}[\textbf{\cite[Definition 1]{jerome}}] \label{MRD definition}
A linear code $ C \subset \mathbb{F}_{q^m}^n $ of dimension $ k $ is $ r $-MRD if $ d_{R,r}(C) = n - k + r $. We say it is MRD if it is $ 1 $-MRD. Similarly for $ r $-MDS and MDS codes, replacing $ d_{R,r} $ by $ d_{H,r} $ (see \cite[Section VI]{wei}).
\end{definition}

We also obtain the bound $ d_{R,r}(C) \leq rm $ from the previous lemma, by induction on $ r $. Therefore, the overview of the Singleton bound becomes now as follows, with notation as above, which improves the bounds in the previous overview:

\begin{displaymath}
d_{R,r}(C) \leq \left\lbrace
\begin{array}{l}
 n - k + r, \\
 rm, \\
 \frac{m}{n}(n-k) + m(r-1) + 1, \\
\end{array}\right. 
\end{displaymath}

\begin{displaymath}
M_{R,s}(C_1,C_2) \leq \left\lbrace \begin{array}{l}
 n - k_1 + s, \\
 sm, \\
 \frac{m(n-k_1)}{n-k_2} + m(s-1) + 1. \\
\end{array} \right.
\end{displaymath}

\begin{remark}
The bound $ d_{R,r}(C) \leq rm $ is sharper than the alternative Singleton bound if, and only if, $ n \geq mk $. We know that in this case, $ C $ is rank degenerate (Proposition \ref{when degenerate}). Therefore, for codes that are not rank degenerate, the usual and alternative Singleton bounds are the sharpest ones.
\end{remark}

\begin{remark}
When $ n \leq m $ the usual Singleton bound is the sharpest general upper bound on the rank distance, since Gabidulin codes (see \cite{gabidulin}) are MRD and may have length $ n $, for all $ n \leq m $, and dimension $ k $, for all $ 1 \leq k \leq n $. 

Since the alternative Singleton bound is sharper for $ r = 1 $ when $ n > m $, it follows immediately that, given $ 1 \leq k \leq n $, and $ m $, there exists an MRD code over $ \mathbb{F}_{q^m}^n $, with length $ n $ and dimension $ k $, if and only if, $ n \leq m $. This gives a result analogous to the MDS conjecture (see \cite[page 265]{pless}) for the rank distance -- although in this case it is not a conjecture.

Also note that the inequality (\ref{r bound}) gives a lower bound on the number $ r $ such that $ C $ is $ r $-MRD.
\end{remark}

\begin{remark}
One might ask if a bound of the form $ d_{R,r}(C) \leq \frac{m}{n}(n-k) + r $ holds, when $ n > m $. However, this is not true even for $ r=2 $. Take for example $ m=2 $, $ n=4 $, $ \alpha \in \mathbb{F}_{q^2} \setminus \mathbb{F}_q $, and the code $ C = \langle (1,\alpha,0,0), (0,0, \alpha, 1) \rangle $, which has dimension $ k =2 $. It is easy to see that $ C^* $ has dimension $ 4 $, since $ (1,\alpha,0,0), (1,\alpha^q,0,0), (0,0, \alpha, 1) $ and $ (0,0, \alpha^q, 1) $ are linearly independent over $ \mathbb{F}_{q^m} $. Thus, for $ r = k = 2 $,
$$ d_{R,2}(C) = 4, \quad \textrm{and} \quad \frac{m}{n}(n-k) + r = \frac{2}{4}(4-2) + 2 = 3. $$
Moreover, we see that $ d_{R,2}(C) $ attains the alternative Singleton bound.
\end{remark}

We conclude the section with a simple fact that connects $ r $-MRD codes with $ r $-MDS codes, and which follows directly from (\ref{def4}).

\begin{proposition}
A linear code $ C \subset \mathbb{F}_{q^m}^n $ is $ r $-MRD if, and only if, $ \varphi_B(C) $ is $ r $-MDS, for all bases $ B \subset \mathbb{F}_q^n $ of $ \mathbb{F}_{q^m}^n $.
\end{proposition}

Thus, if $ C $ is a Gabidulin code \cite{gabidulin}, it is obviously MDS, but also the codes $ \varphi_B(C) $ are MDS. It can also be easily shown that the codes $ \varphi_B(C) $ are again Gabidulin codes. Therefore, to prove that they are MRD, it is only necessary to prove that they are MDS.

\section{Rank-puncturing and rank-shortening} \label{sec puncturing}

In this section we discuss what are the operations on rank-metric codes analogous to puncturing and shortening \cite[Section 1.5]{pless}. The main importance of the concept of puncturing is that a punctured codeword is essentially the same as a codeword with erasures, as in the Hamming case. Recall that the shortened and punctured codes of a given code $ C \subset \mathbb{F}_{q^m}^n $ on the coordinates in the set $ I \subset \mathcal{J} $ are defined, respectively, as
\begin{equation*}
\begin{split}
C_I = & C \cap V_I = \{ \mathbf{c} \in C \mid c_i = 0, \forall i \notin I \}, \\
C^I = & \{ (c_i)_{i \in I} \mid \mathbf{c} \in C \}.
\end{split}
\end{equation*}

\subsection{The definitions}

For a linear subspace $ L \subset \mathbb{F}_q^n $, fix another subspace $ L^\prime \subset \mathbb{F}_q^n $ such that $ \mathbb{F}_q^n = L^\prime \oplus L^\perp $. Observe that $ \dim(L) = n - \dim(L^\perp) = \dim(L^\prime) $, which we will use throughout the section. We then define the projection map 
$$ \pi_{L,L^\prime} : \mathbb{F}_{q^m}^n \longrightarrow V^\prime = L^\prime \otimes \mathbb{F}_{q^m}, $$
such that $ \pi_{L,L^\prime}(\mathbf{c}) = \mathbf{c}_1 $, where $ \mathbf{c} = \mathbf{c}_1 + \mathbf{c}_2 $, $ \mathbf{c}_1 \in V^\prime = L^\prime \otimes \mathbb{F}_{q^m} $ and $ \mathbf{c}_2 \in V^\perp = L^\perp \otimes \mathbb{F}_{q^m} $. We then write $ C^{L,L^\prime} = \pi_{L,L^\prime}(C) $, for an (arbitrary) code $ C \subset \mathbb{F}_{q^m}^n $.

\begin{lemma}
For any two subspaces $ L^\prime, L^{\prime \prime} \subset \mathbb{F}_q^n $ such that $ \mathbb{F}_q^n = L^\prime \oplus L^\perp = L^{\prime \prime} \oplus L^\perp $, and for any code $ C \subset \mathbb{F}_{q^m}^n $, we have that the codes $ C^{L,L^\prime} $ and $ C^{L,L^{\prime \prime}} $ are rank-metric equivalent in a canonical way. 
\end{lemma}
\begin{proof}
Define $ \phi : V^\prime \longrightarrow V^{\prime \prime} $ by $ \phi (\mathbf{c}) = \pi_{L,L^{\prime \prime}}(\mathbf{c}) $, where $ V^\prime = L^\prime \otimes \mathbb{F}_{q^m} $ and $ V^{\prime \prime} = L^{\prime \prime} \otimes \mathbb{F}_{q^m} $. 

First we see that $ \phi $ is a vector space isomorphism. Since $ \dim(V^\prime) = \dim(V^{\prime \prime}) $, we only need to prove that it is one to one. Assume that $ \pi_{L,L^{\prime \prime}}(\mathbf{c}) = \mathbf{0} $. This means that $ \mathbf{c} \in V^\perp $, but also $ \mathbf{c} \in V^\prime $ and $ V^\prime \cap V^\perp = 0 $, hence $ \mathbf{c} = \mathbf{0} $.

On the other hand, since $ \mathbb{F}_q^n = L^{\prime \prime} \oplus L^\perp $, if $ \mathbf{c} \in L^\prime $, then $ \phi(\mathbf{c}) \in L^{\prime \prime} $. In other words, $ \phi(V^\prime \vert_{\mathbb{F}_q}) \subset V^{\prime \prime} \vert_{\mathbb{F}_q} $. By Theorem \ref{equiv rank}, item 5, $ \phi $ is a rank-metric equivalence.

Finally, we see that $ \phi (C^{L,L^\prime}) = C^{L,L^{\prime \prime}} $. If $ \mathbf{c}_1 \in C^{L,L^\prime} $, then there exists $ \mathbf{c} = \mathbf{c}_1 + \mathbf{c}_2 \in C $, with $ \mathbf{c}_2 \in V^\perp $. Write $ \mathbf{c} = \widetilde{\mathbf{c}}_1 + \widetilde{\mathbf{c}}_2 $, with $ \widetilde{\mathbf{c}}_1 \in V^{\prime \prime} $ and $ \widetilde{\mathbf{c}}_2 \in V^\perp $. Then $ \mathbf{c}_1 = \widetilde{\mathbf{c}}_1 + (\widetilde{\mathbf{c}}_2 - \mathbf{c}_2) $ and hence $ \phi(\mathbf{c}_1) = \widetilde{\mathbf{c}}_1 \in C^{L,L^{\prime \prime}} $.
\end{proof}

Therefore, the next definition of rank-punctured code is consistent.

\begin{definition} \label{punctured and shortened}
For every $ \mathbb{F}_q $-linear space $ L \subset \mathbb{F}_q^n $, and every code $ C \subset \mathbb{F}_{q^m}^n $, we define its rank-punctured and rank-shortened codes over $ L $ as $ C^L = C^{L,L^\prime} $ and $ C_L = C \cap V $, respectively, for some $ L^\prime $ as before, where $ V = L \otimes \mathbb{F}_{q^m} $.

Similarly, for a coding scheme $ \mathcal{P}_\mathcal{S} = \{ C_{\mathbf{x}} \}_{\mathbf{x} \in \mathcal{S}} $, we can define its rank-punctured and rank-shortened schemes over $ L $ as $ \mathcal{P}_\mathcal{S}^L = \{ C_{\mathbf{x}}^L \}_{\mathbf{x} \in \mathcal{S}} $ and $ {\mathcal{P}_\mathcal{S}}_L = \{ C_{\mathbf{x} L} \}_{\mathbf{x} \in \mathcal{S}} $, respectively. For a linear coding scheme built from $ C_2 \varsubsetneq C_1 \subset \mathbb{F}_{q^m}^n $, they are the schemes built from $ C_2^L \subset C_1^L $ and $ C_{2L} \subset C_{1L} $, respectively.
\end{definition}

Observe that it is not always true that $ C_L \subset C^L $, as opposed to the usual shortening and puncturing. On the other hand, we see that, for every $ I \subset \mathcal{J} $, $ V_I \in \Upsilon $. Then, it is easy to see that $ C^I = C^{L_I} $ and $ C_I = C_{L_I} $, regarded as subspaces of $ V_I $. Thus the previous definition extends the usual definition of puncturing and shortening. For brevity, we will use just the words puncturing and shortening for rank-puncturing and rank-shortening, respectively.

\begin{remark}
Note that, given $ L \subset \mathbb{F}_q^n $, there may be more than one subspace $ L^\prime \subset \mathbb{F}_q^n $ such that $ \mathbb{F}_q^n = L^\prime \oplus L^\perp $ (later we will actually see how to obtain them). If $ V = L \otimes \mathbb{F}_{q^m} $, then $ V^\perp = L^\perp \otimes \mathbb{F}_{q^m} $, and what we are doing is finding a subspace $ V^\prime \in \Upsilon $ such that $ \mathbb{F}_{q^m}^n = V^\prime \oplus V^\perp $.

On the other hand, if $ V = V_I \in \Lambda $, then $ V_I^\perp = V_{\overline{I}} $ and $ V_I $ is the unique subspace $ V^\prime \in \Lambda $ such that $ \mathbb{F}_{q^m}^n = V^\prime \oplus V^\perp $. Therefore, punctured codes in the Hamming case are defined in a unique way, in contrast with the rank case.
\end{remark}

Usually, $ C^I $ and $ C_I $ are considered as subspaces of $ \mathbb{F}_{q^m}^{\# I} $. This is obvious since $ {\rm Supp}(C^I) \subset I $ and $ V_I $ is Hamming equivalent to $ \mathbb{F}_{q^m}^{\# I} $. For rank-metric codes, we can fix bases $ B, B^\prime $ of $ L, L^\prime \subset \mathbb{F}_q^n $, respectively, and consider $ \psi_B (C_L) $ and $ \psi_{B^\prime} (C^L) $, where $ \psi_B $ and $ \psi_{B^\prime} $ are as in (\ref{psi basis}). That is, we can consider that $ C_L, C^L \subset \mathbb{F}_{q^m}^{\dim(L)} $.

\subsection{$ r $-MRD characterizations}

In this subsection, we give characterizations of $ r $-MRD (and $ r $-MDS) codes in terms of dimensions of punctured codes. We start with a tool that generalizes Forney's Lemmas \cite[Lemmas 1 and 2]{forney} and that is useful to relate dimensions of punctured and shortened codes. Note that \cite[Lemma 10]{rgrw} is essentially the second equality in this lemma.

\begin{lemma} \label{forney}
For every linear code $ C \subset \mathbb{F}_{q^m}^n $ of dimension $ k $ and every subspace $ L \subset \mathbb{F}_q^n $, it holds that
$$ \dim(C^L) = \dim(L) - \dim((C^\perp)_L) = k - \dim(C_{L^\perp}). $$
\end{lemma}
\begin{proof}
The second equality is \cite[Lemma 10]{rgrw}. Now $ \dim (C^L) = \dim (\pi_{L,L^\prime}(C)) = k - \dim ({\rm ker}(\pi_{L,L^\prime})) = k - \dim (C_{L^\perp}) $.
\end{proof}

We will need the duality theorem for generalized rank weights, which has been established and proven in \cite{jerome} (we will give a shorter proof in Appendix \ref{app 2}):

\begin{theorem}[\textbf{Duality \cite{jerome}}] \label{wei duality}
Given a linear code $ C \subset \mathbb{F}_{q^m}^n $ of dimension $ k $, write $ d_r = d_{R,r}(C) $ for $ 1 \leq r \leq k $, and $ d_s^\perp = d_{R,s}(C^\perp) $, for $ 1 \leq s \leq n-k $. Then it holds that
\begin{equation*}
\begin{split}
\{ 1,2, \ldots, n \} = & \{ d_1, d_2, \ldots, d_k \} \cup \\
 & \{ n+1-d_1^\perp, n+1 -d_2^\perp, \ldots, n+1- d_{n-k}^\perp \},
\end{split}
\end{equation*}
where the union is disjoint.
\end{theorem}

Note that, in the next propositions, the equivalence of the two first conditions follows directly from Wei's duality and its corresponding theorem for rank weights, as proven in \cite[Proposition 4.1]{tsfasman} and \cite[Corollary III.3]{jerome}, respectively. The equivalence between item 2 and item 4 for Hamming weights is proven in \cite[Theorem 1.4.15]{pless}, and the case $ r = 1 $ ($ C $ is MDS) is fully proven in \cite[Theorem 2.4.3]{pless}. It also generalizes \cite[Corollary 1.4.14]{pless} and \cite[Theorem 1.5.7 (ii)]{pless}.

\begin{proposition}  \label{r MDS}
The following conditions are equivalent for a linear code $ C \subset \mathbb{F}_{q^m}^n $ of dimension $ k $, and every $ 1 \leq r \leq k $:
\begin{enumerate}
\item
The code $ C $ is $ r $-MDS.
\item
$ d_{H,1}(C^\perp) \geq k - r + 2 $.
\item
For all $ I \subset \mathcal{J} $ such that $ \# I \leq k - r + 1 $, we have that $ \dim (C^I) = \# I $.
\item
For all $ I \subset \mathcal{J} $ such that $ \# I \geq n - k + r - 1 $, we have that $ \dim ((C^\perp)^I) = n-k $.
\end{enumerate}
\end{proposition}

\begin{proposition}  \label{r MRD}
The following conditions are equivalent for a linear code $ C \subset \mathbb{F}_{q^m}^n $ of dimension $ k $, and every $ 1 \leq r \leq k $:
\begin{enumerate}
\item
The code $ C $ is $ r $-MRD.
\item
$ d_{R,1}(C^\perp) \geq k - r + 2 $.
\item
For all $ L \subset \mathbb{F}_q^n $ such that $ \dim(L) \leq k - r + 1 $, we have that $ \dim (C^L) = \dim(L) $.
\item
For all $ L \subset \mathbb{F}_q^n $ such that $ \dim(L) \geq n - k + r - 1 $, we have that $ \dim ((C^\perp)^L) = n-k $.
\end{enumerate}
\end{proposition}
\begin{proof}
The equivalence between the first two conditions follows from the duality Theorem \ref{wei duality}, as proven in \cite{jerome}, and the equivalence between the last two conditions follows from Lemma \ref{forney}.

Now, we prove that condition 3 implies condition 2. Take $ \mathbf{c} \in C^\perp \setminus 0 $ and assume that $ {\rm wt_R}(\mathbf{c}) = \dim(L) \leq k - r +1 $, where $ L = (\langle \mathbf{c} \rangle^*) \vert_{\mathbb{F}_q} $ (recall $ {\rm wt_R}(\mathbf{c}) = \dim(\langle \mathbf{c} \rangle^*) $ from Lemma \ref{charact stic}). Then by Lemma \ref{forney},
$$ \dim(L) = \dim (C^L) = \dim(L) - \dim ((C^\perp)_L), $$
and thus $ (C^\perp)_L = 0 $, but this implies that $ \mathbf{c} = \mathbf{0} $, which is a contradiction. Hence $ {\rm wt_R} (\mathbf{c}) \geq k - r +2 $.

Finally, we prove that condition 2 implies condition 3. Let $ L \subset \mathbb{F}_q^n $ be such that $ \dim(L) \leq k - r + 1 $. Then, by the definition of minimum rank distance (recall (\ref{def1})), we have that $ \dim ((C^\perp)_L) = 0 $, and thus by Lemma \ref{forney},
$$ \dim(C^L) = \dim(L) - \dim ((C^\perp)_L) = \dim(L). $$
\end{proof}

After showing how to compute generator matrices for punctured codes, it can be easily proven that the equivalence between items 2 and 3 generalizes \cite[Theorem 1]{gabidulin}.

\begin{corollary}
The smallest integer $ r $ such that $ C $ is $ r $-MDS is $ r = k - d_{H,1}(C^\perp) + 2 $, and similarly for rank weights.
\end{corollary}

\subsection{Information spaces}

Next, we define the notion of information space, which plays the same role as information sets in the Hamming case: any original codeword can be recovered from the punctured codeword if (and also only if in the linear case) we puncture on an information space. Therefore, information spaces completely describe the erasure correction capability of a code, and not only worst cases.

\begin{definition} \label{info spaces}
Given a linear code $ C \subset \mathbb{F}_{q^m}^n $, we say that a subspace $ L \subset \mathbb{F}_q^n $ is an information space for $ C $ if $ \dim(C^L) = \dim (C) $. Equivalently, if the restriction $ \pi_{L,L^\prime} : C \longrightarrow C^L $ is an $ \mathbb{F}_{q^m} $-linear vector space isomorphism.

For an (arbitrary) code $ C \subset \mathbb{F}_{q^m}^n $, we say that $ L $ is an information space for  $ C $ if $ \pi_{L,L^\prime} : C \longrightarrow C^L $ is bijective.

On the other hand, given a code pair $ C_2 \varsubsetneq C_1 \subset \mathbb{F}_{q^m}^n $, we say that $ L $ is an information space for $ C_1, C_2 $ if $ \dim(C_1^L/C_2^L) = \dim(C_1/C_2) $. In general, for an (arbitrary) coding scheme $ \mathcal{P}_\mathcal{S} = \{ C_{\mathbf{x}} \}_{\mathbf{x} \in \mathcal{S}} $, we say that $ L $ is an information space for $ \mathcal{P}_\mathcal{S} $ if $ \pi_{L,L^\prime}(C_{\mathbf{x}_1}) \cap \pi_{L,L^\prime}(C_{\mathbf{x}_2}) = \varnothing $, whenever $ \mathbf{x}_1 \neq \mathbf{x}_2 $.
\end{definition} 

Observe that a set $ I \subset \mathcal{J} $ is an information set for $ C $ if, and only if, $ L_I $ is an information space for $ C $. Note also that $ \pi_{L,L^\prime} $ is always surjective, so it is only necessary to be injective in order to be bijective.

On the other hand, Proposition \ref{r MRD}, item 4, shows threshold values on the dimension of a space to guarantee that it is an information space for a given code, in terms of its minimum rank distance, as in the Hamming case.

Now we characterize MRD codes using information spaces, in the same way as MDS codes are characterized using information sets. Note that the result is a particular case of Proposition \ref{r MRD}, taking $ r=1 $. After knowing how to compute generator matrices of punctured codes, it can be shown that this proposition is essentially \cite[Theorem 2]{gabidulin}.

\begin{proposition} \label{info spaces def}
A linear code $ C \subset \mathbb{F}_{q^m}^n $ is MRD if, and only if, every $ L \subset \mathbb{F}_q^n $, with $ \dim(L) = k = \dim(C) $, is an information space for $ C $.
\end{proposition}

The following two propositions essentially describe erasure correction on networks. The second one also describes the correction capability of punctured codes. They are analogous to \cite[Theorem 1.5.7 (ii)]{pless} and \cite[Theorem 1.5.1]{pless}, respectively. The first one also extends \cite[Theorem 1]{gabidulin} to arbitrary codes.

\begin{proposition} \label{info spaces distance 1}
Given an (arbitrary) code $ C \subset \mathbb{F}_{q^m}^n $, if $ \rho < d_R(C) $, then every subspace $ L \subset \mathbb{F}_q^n $ with $ \dim (L) \geq n - \rho $ is an information space for $ C $. If $ \rho \geq d_R(C) $, there exists a subspace $ L \subset \mathbb{F}_q^n $ with $ \dim (L) = n - \rho $ which is not an information space for $ C $.
\end{proposition}
\begin{proof}
First we prove in the first case that $ \pi_{L,L^\prime} : C \longrightarrow C^L $ is injective. Take $ \mathbf{c}_1, \mathbf{c}_2 \in C $ such that $ \pi_{L,L^\prime} (\mathbf{c}) = \mathbf{0} $, where $ \mathbf{c} = \mathbf{c}_1 - \mathbf{c}_2 $. Then, $ \mathbf{c} \in V^\perp $, $ V = L \otimes \mathbb{F}_{q^m} $, and therefore, $ {\rm wt_R} (\mathbf{c}) \leq \dim (V^\perp) \leq \rho $, which is absurd.

For the second statement, take $ \mathbf{c}_1, \mathbf{c}_2 $ and $ \mathbf{c} = \mathbf{c}_1 - \mathbf{c}_2 $ such that $ {\rm wt_R}(\mathbf{c}) = d_R (C) $, write $ D = \langle \mathbf{c} \rangle ^* = \langle \mathbf{v}_1, \mathbf{v}_2, \ldots, \mathbf{v}_s \rangle $, with $ \mathbf{v}_i \in \mathbb{F}_q^n $, and extend this to a basis $ B = \{ \mathbf{v}_i \}_{i=1}^n $ of $ \mathbb{F}_q^n $. Consider $ L^\perp = \langle \mathbf{v}_1, \mathbf{v}_2, \ldots, \mathbf{v}_\rho \rangle_{\mathbb{F}_q} $, then $ \dim (L) = n - \rho $ and $ \pi_{L,L^\prime}(\mathbf{c}_1) = \pi_{L,L^\prime}(\mathbf{c}_2) $.
\end{proof}

\begin{proposition} \label{info spaces distance 2}
Given an (arbitrary) code $ C \subset \mathbb{F}_{q^m}^n $ with $ \rho < d_R(C) $, every subspace $ L \subset \mathbb{F}_q^n $ with $ \dim (L) \geq n - \rho $ satisfies that $ d_R(C^L) \geq d_R(C) - \rho $. Moreover, there exists a subspace $ L \subset \mathbb{F}_q^n $ with $ \dim (L) = n - \rho $ such that $ d_R (C^L) = d_R (C) - \rho $.
\end{proposition}
\begin{proof}
With the same notation as in the previous proof, we have that $ {\rm wt_R}(\pi_{L,L^\prime} (\mathbf{c})) = \dim ( \langle \pi_{L,L^\prime} (\mathbf{c}) \rangle ^*) \geq \dim (\langle \mathbf{c} \rangle ^*) - \rho $, and the first statement follows.

Finally, take $ \mathbf{c}_1, \mathbf{c}_2 $ such that $ {\rm wt_R}(\mathbf{c}) = d_R (C) $, and write $ D = \langle \mathbf{c} \rangle ^* = \langle \mathbf{v}_1, \mathbf{v}_2, \ldots, \mathbf{v}_s \rangle $, with $ \mathbf{v}_i \in \mathbb{F}_q^n $, and extend this to a basis $ B = \{ \mathbf{v}_i \}_{i=1}^n $ of $ \mathbb{F}_q^n $. Consider $ L^\perp = \langle \mathbf{v}_1, \mathbf{v}_2, \ldots, \mathbf{v}_\rho \rangle_{\mathbb{F}_q} $ and $ L^\prime = \langle \mathbf{v}_{\rho +1}, \mathbf{v}_{\rho + 2}, \ldots, \mathbf{v}_n \rangle_{\mathbb{F}_q} $, then $ {\rm ker}(\pi_{L,L^\prime}) \cap D = L^\perp \otimes \mathbb{F}_{q^m} $, and therefore $ {\rm wt_R}(\pi_{L,L^\prime}(\mathbf{c})) = {\rm wt_R}(\mathbf{c}) - \rho $, and the last statement follows.
\end{proof}

We can extend this to (arbitrary) coding schemes, just by substituting the code $ C $ with a coding scheme $ \mathcal{P}_\mathcal{S} = \{ C_{\mathbf{x}} \}_{\mathbf{x} \in \mathcal{S}} $. The proof is the same.

\subsection{Computing rank-punctured codes}

We conclude the section showing how to compute punctured codes. In the Hamming case, this is obvious, since we only have to project on some of the coordinates. In the rank case, we need to solve some systems of linear equations, which is still an efficient computation.

\begin{proposition}
Given a subspace $ L \subset \mathbb{F}_q^n $ and one of its generator matrices $ A $ ($ L = {\rm row}(A) $ and $ A $ has full rank \cite{pless}), we have that a subspace $ L^\prime \subset \mathbb{F}_q^n $ satisfies $ \mathbb{F}_q^n = L^\prime \oplus L^\perp $ if, and only if, it has a generator matrix $ A^\prime $ such that $ A^\prime A^T = I $.
\end{proposition}
\begin{proof}
First assume that $ \mathbb{F}_q^n = L^\prime \oplus L^\perp $ and $ B $ is a generator matrix for $ L^\prime $. Take $ \mathbf{x} $ such that $ \mathbf{x} B A^T = \mathbf{0} $, then $ \mathbf{x} B \in L^\prime \cap L^\perp $ and therefore, $ \mathbf{x} B = \mathbf{0} $, which implies that $ \mathbf{x} = \mathbf{0} $. Hence, $ BA^T $ is full rank and there exists an invertible matrix $ M $ such that $ MBA^T = I $. Taking $ A^\prime = MB $ we obtain the desired matrix.

Now assume that $ L^\prime $ has a generator matrix $ A^\prime $ with $ A^\prime A^T = I $. Since $ \dim(L^\prime) = \dim(L) = n - \dim(L^\perp) $, we need to prove that $ L^\prime \cap L^\perp = 0 $. Suppose that $ \mathbf{x} A^\prime \in L^\perp $, then $ \mathbf{x} = \mathbf{x} A^\prime A^T = \mathbf{0} $, and we are done.
\end{proof}

Therefore, to compute subspaces $ L^\prime $ with $ \mathbb{F}_q^n = L^\prime \oplus L^\perp $, we just need to solve the equations $ A \mathbf{a}^{\prime T}_i = \mathbf{e}^T_i $, $ i=1,2, \ldots, \dim(L) $. Different solutions give different spaces.

Note that if $ A $ is a generator matrix of $ L \subset \mathbb{F}_q^n $ over $ \mathbb{F}_q $, then it is a generator matrix of $ V = L \otimes \mathbb{F}_{q^m} $ over $ \mathbb{F}_{q^m} $.

\begin{lemma} \label{lemma projection}
With the same notation as in the previous proposition, we have that, for every $ \mathbf{c} \in \mathbb{F}_{q^m}^n $, 
$$ \pi_{L,L^\prime}(\mathbf{c}) = \mathbf{c} A^TA^\prime. $$
\end{lemma}

And now we give a method to compute the generator matrix of a punctured code $ C^L $, given generator matrices of $ C $ and $ L $. The proof is straightforward and follows from the previous lemma.

\begin{proposition} \label{generator punctured}
Let $ C \subset \mathbb{F}_{q^m}^n $ be a linear code with generator matrix $ G $, and let $ L,L^\prime \subset \mathbb{F}_q^n $ be subspaces with generator matrices $ A $ and $ A^\prime $, respectively, and such that $ A^\prime A^T = I $.

We have that $ GA^TA^\prime $ satisfies that $ {\rm row}(GA^TA^\prime) = C^{L,L^\prime} = C^L $, and thus by deleting linearly dependent rows, we obtain a generator matrix for $ C^L $. Moreover, if $ L $ is an information space for $ C $, then $ G A^T A^\prime $ is full rank and therefore it is a generator matrix for $ C^L $.
\end{proposition}

\section{Secure network coding} \label{sec secure}

In this section we revisit the description of secure linear network coding in view of the results in the previous sections. Recall from Subsection \ref{subsec linear network} the linear network coding with errors that we are considering, which is the one in \cite{rgrw, on-metrics}, and recall from Subsection \ref{subsec coding schemes} that we assume that the source encodes the original message $ \mathbf{x} \in \mathbb{F}_{q^m}^k $ into $ \mathbf{c} \in \mathbb{F}_{q^m}^n $ using some coding scheme $ \mathcal{P}_\mathcal{S} = \{ C_\mathbf{x} \}_{\mathbf{x} \in \mathcal{S}} $.

As explained in the introduction, we consider an adversary that may compromise the security of the network by doing three things: introducing $ t $ erroneous packets on $ t $ different links, modifying the transfer matrix $ A $ and obtaining information about the original message $ \mathbf{x} $ by wiretapping several links. 

As in \cite{rgrw, on-metrics}, if the receiver obtains the vector $ \mathbf{y} = \mathbf{c} A^T + \mathbf{e} $, $ t = {\rm wt_R}(\mathbf{e}) $ and $ \rho = n - {\rm Rk}(A) $, then we say that $ t $ errors and $ \rho $ erasures occurred. In Appendix \ref{app 3}, we will see how to consider erasures as errors.

\subsection{Erasure correction and information leakage revisited} \label{erasure sect}

In this subsection we study the problems of erasure correction and information leakage, which are closely related. The amount of leaked information on networks was studied in \cite{rgrw}. We will see how the punctured construction in Section \ref{sec puncturing} can describe this.

Consider a linear coding scheme built from $ C_2 \varsubsetneq C_1 \subset \mathbb{F}_{q^m}^n $. Denote by $ S $ and $ X $ the random variables corresponding to the original message and the encoded message by the previous nested coset coding scheme, respectively, and $ \pi_I $ the projection onto the coordinates in $ I \subset \mathcal{J} $. It was shown in \cite{one-point} and \cite{kurihara-secret} that 
\begin{equation} \label{info leakage punctured 1}
{\rm I}(S;\pi_I(X)) = \dim((C_2^\perp)_I/(C_1^\perp)_I) = \dim(C_1^I / C_2^I),
\end{equation}
for every $ I \subset \mathcal{J} $, assuming a uniform distribution, where the last equality follows from Lemma \ref{forney}, and $ {\rm I}(X;Y) = H(X) - H(X | Y) $ is the mutual information of the random variables $ X $ and $ Y $.

On the other hand, by wiretapping $ s $ links in a network, an adversary obtains the variable $ X B^T $, for some matrix $ B \in \mathbb{F}_q^{s \times n} $. Assuming uniform distributions, and defining $ L = {\rm row}(B) \subset \mathbb{F}_q^n $, it is proven in \cite[Lemma 7]{rgrw} that
\begin{equation} \label{info leakage punctured 2}
{\rm I}(S; X B^T) = \dim((C_2^\perp)_L/(C_1^\perp)_L) = \dim(C_1^L / C_2^L),
\end{equation}
where the last equality follows from Lemma \ref{forney}. 

Therefore, the information leakage is tightly related to the dimension of punctured and shortened codes. 

Observe that $ {\rm I}(S; X B^T) \leq \dim(C_1/C_2) $ and the equality holds if, and only if, $ L $ is an information space for $ C_1, C_2 $ as in Definition \ref{info spaces}. Remember from Proposition \ref{info spaces distance 1} that if $ n - {\rm Rk}(B) < d_R(\mathcal{P}_\mathcal{S}) $, then $ L = {\rm row}(B) $ is an information space for $ \mathcal{P}_\mathcal{S} $. In Appendix \ref{app 1}, we show how to efficiently obtain the original message if $ L $ is an information space. \\

Next we give a relation between information leakage and duality, whose philosophy is similar to that of MacWilliams equations, since it means that knowing the information leakage using the code pair $ C_2 \varsubsetneq C_1 $ is equivalent to knowing the information leakage using the ``dual'' code pair $ C_1^\perp \varsubsetneq C_2^\perp $. It is convenient to introduce the definition of access structures:

\begin{definition}[\textbf{\cite{one-point}}] \label{access structures 1}
We define the Hamming access structure of the nested linear code pair $ C_2 \varsubsetneq C_1 $ as the collection of the following sets
$$ \mathcal{A}(C_1,C_2)_r = \{ I \subset \mathcal{J} \mid \dim(C_1^I / C_2^I) = r \}, $$
for $ 0 \leq r \leq \ell= \dim(C_1/C_2) $. Given a set $ \mathcal{A} \subset \mathcal{P}(\mathcal{J}) $, we define its Hamming dual as $ \mathcal{A}^\perp = \{ I \subset \mathcal{J} \mid \overline{I} \in \mathcal{A} \} $.
\end{definition}

\begin{definition} \label{access structures 2}
We define the rank access structure of the nested linear code pair $ C_2 \varsubsetneq C_1 $ as the collection of the following linear subspaces of $ \mathbb{F}_q^n $
$$ \mathcal{B}(C_1,C_2)_r = \{ L \subset \mathbb{F}_q^n \mid \dim(C_1^L / C_2^L) = r \}, $$
for $ 0 \leq r \leq \ell= \dim(C_1/C_2) $. Given a set $ \mathcal{B} \subset \{ L \subset \mathbb{F}_q^n \textrm{ linear subspace} \} $, we define its rank dual as $ \mathcal{B}^\perp = \{ L \subset \mathbb{F}_q^n \mid L^\perp \in \mathcal{B} \} $.
\end{definition}

We now present the relation with duality, where the Hamming case for $ r = 0 $ was already proven in \cite[Proof of Theorem 1]{secure-computation} for the Massey-type scheme \cite[Section 3]{secure-computation}. The rank case and the general Hamming case are new.

\begin{proposition} 
Given a nested linear code pair $ C_2 \varsubsetneq C_1 \subset \mathbb{F}_{q^m}^n $ and $ 0 \leq r \leq \ell = \dim(C_1/C_2) $, we have that
$$ \mathcal{A}(C_2^\perp,C_1^\perp)_r = \mathcal{A}(C_1,C_2)_{\ell-r}^\perp. $$
\end{proposition}
\begin{proof}
It follows from the following equality, which follows from Lemma \ref{forney}, 
$$ \dim((C_2^\perp)^I / (C_1^\perp)^I) + \dim(C_1^{\overline{I}} / C_2^{\overline{I}}) = \ell. $$
\end{proof}

\begin{proposition} \label{pre duality}
Given a nested linear code pair $ C_2 \varsubsetneq C_1 \subset \mathbb{F}_{q^m}^n $ and $ 0 \leq r \leq \ell = \dim(C_1/C_2) $, we have that
$$ \mathcal{B}(C_2^\perp,C_1^\perp)_r = \mathcal{B}(C_1,C_2)_{\ell-r}^\perp. $$
\end{proposition}
\begin{proof}
Again, it follows from the following equality, which follows from Lemma \ref{forney},
$$ \dim((C_2^\perp)^L / (C_1^\perp)^L) + \dim(C_1^{L^\perp} / C_2^{L^\perp}) = \ell. $$ \\
\end{proof}

Finally, as consequences of Proposition \ref{r MDS} and Proposition \ref{r MRD}, we obtain the description of the access structures for MDS and MRD code pairs, respectively. The Hamming case (Corollary \ref{computing info leakage MDS}) also follows immediately from \cite[Section III]{one-point}.

\begin{corollary}[\textbf{\cite[Section III]{one-point}}] \label{computing info leakage MDS}
If both $ C_1 $ and $ C_2 $ are MDS, then
\begin{displaymath}
\dim(C_1^I / C_2^I) = \left\lbrace 
\begin{array}{ll}
 \ell &, \textrm{ if } k_1 \leq \# I, \\
 \# I - k_2 &, \textrm{ if } k_2 \leq \# I \leq k_1, \\
 0 &, \textrm{ if } \# I \leq k_2, \\
\end{array} 
\right.
\end{displaymath}
for every $ I \subset \mathcal{J} $.
\end{corollary}

\begin{corollary}
If both $ C_1 $ and $ C_2 $ are MRD, then
\begin{displaymath}
\dim(C_1^L / C_2^L) = \left\lbrace 
\begin{array}{ll}
 \ell &, \textrm{ if } k_1 \leq \dim(L), \\
 \dim(L) - k_2 &, \textrm{ if } k_2 \leq \dim(L) \leq k_1, \\
 0 &, \textrm{ if } \dim(L) \leq k_2, \\
\end{array} 
\right.
\end{displaymath}
for every linear subspace $ L \subset \mathbb{F}_q^n $.
\end{corollary}

In general, we can compute the information leaked in many cases, but if the involved codes are not MDS (respectively, MRD), then there is always a collection of sets (respectively, subspaces) for which we do not completely know the information leaked. We first establish this fact for the rank case, which follows from Proposition \ref{r MRD}, and give an example in the Hamming case:

\begin{proposition}
Let $ C_2 \varsubsetneq C_1 \subset \mathbb{F}_{q^m}^n $ be a nested linear code pair such that $ k_i = \dim(C_i) $, $ i=1,2 $, $ \ell = k_1 - k_2 $, $ C_1 $ is $ r_1 $-MRD and $ C_2^\perp $ is $ r_2 $-MRD, or equivalently, $ d_R(C_1^\perp) \geq k_1 - r_1 + 2 $ and $ d_R(C_2) \geq n - k_2 - r_2 + 2 $. If $ L \subset \mathbb{F}_q^n $ is a subspace such that $ k_2 + r_2 - 1 \leq \dim(L) \leq k_1 - r_1 + 1 $, then $ \dim(C_1^L / C_2^L) = \dim(L) - k_2 $, which only depends on $ \dim(L) $ and not on the space $ L $.

If moreover, $ k_2 + r_2 - 1 < k_1 - r_1 + 1 $, and taking $ s_1 = n - k_1 - d(C_1) + 1 $ and $ s_2 = k_2 - d(C_2^\perp) + 1 $, then for every subspace $ L \subset \mathbb{F}_q^n $, it holds that $ \dim(C_1^L / C_2^L) $ is
\begin{displaymath}
\begin{array}{ll}
= \ell &, \textrm{ if } k_1 + s_1 \leq \dim(L), \\
\geq \ell - r_1 + 1 &, \textrm{ if } k_1 - r_1 + 1 < \dim(L) < k_1 + s_1, \\
= \dim(L) - k_2 &, \textrm{ if } k_2 + r_2 - 1 \leq \dim(L) \leq k_1 - r_1 + 1, \\
\leq r_2 - 1 &, \textrm{ if } k_2 - s_2 < \dim(L) < k_2 + r_2 - 1, \\
= 0 &, \textrm{ if } \dim(L) \leq k_2 - s_2. \\
\end{array}
\end{displaymath}
\end{proposition}

\begin{example}
If $ C_1 $ and $ C_2 $ are algebraic geometric codes constructed from a function field of genus $ g $ \cite{tsfasman}, then we have the Goppa bound \cite[Theorem 4.3]{tsfasman}: $ d_{H,1}(C_i) \geq n - \dim(C_i) + 1 - g $ and $ d_{H,1}(C_i^\perp) \geq \dim(C_i) + 1 - g $. It follows from Proposition \ref{r MDS} that, for the code pair $ C_2 \varsubsetneq C_1 $,
\begin{displaymath}
\dim(C_1^I / C_2^I) \left\lbrace 
\begin{array}{ll}
= \ell &, \textrm{ if } k_1 + g \leq \# I, \\
\geq \ell - g &, \textrm{ if } k_1 - g < \# I < k_1 + g, \\
= \# I - k_2 &, \textrm{ if } k_2 + g \leq \# I \leq k_1 - g, \\
\leq g &, \textrm{ if } k_2 - g < \# I < k_2 + g, \\
= 0 &, \textrm{ if } \# I \leq k_2 - g. \\
\end{array}
\right.
\end{displaymath}
\end{example}

\subsection{Error and erasure correction revisited}

In this subsection we see how the rank-puncturing can describe error and erasure correction in networks. We will follow a slightly different approach than that of \cite{rgrw, on-metrics}. 

We will treat the coherent case, that is, the case in which the matrix $ A $ is known by the receiver. For simplicity, we will consider the case of one code $ C \subset \mathbb{F}_{q^m}^n $, which may be non-linear. At the end we will show how to adapt the results to arbitrary coding schemes. Observe that \cite[Theorem 4]{rgrw} only deals with linear (meaning $ \mathbb{F}_{q^m} $-linear, as in the rest of the paper) coding schemes.

As we saw in the previous subsection (see also Appendix \ref{app 1}), if the sink node receives $ \mathbf{y} = \mathbf{c} A^T $ and the number of erasures is less than $ d_R(C) $, we can perform erasure correction. For that, we can take a submatrix $ \widetilde{A} $ of $ A $ which is a generator matrix of $ L = {\rm row}(A) $, since the other rows in $ A $ are redundant. All choices of $ \widetilde{A} $ will give the same unique solution.

When there are errors, we would also like to take a submatrix as before and the corresponding subvector of $ \mathbf{y} $. However, it is not clear that the decoder in \cite{rgrw, on-metrics} for $ A $ and for $ \widetilde{A} $ will behave in the same way. We now propose a slightly different approach.

Fix the positive integer $ N $ and the matrix $ A \in \mathbb{F}_q^{N \times n} $, which are assumed to be known by the receiver. 

\begin{definition}[\textbf{\cite[Equations (9), (12)]{on-metrics}}] \label{def discrepancy}
For each $ \mathbf{c} \in \mathbb{F}_{q^m}^n $ and $ \mathbf{y} \in \mathbb{F}_{q^m}^N $, we define the discrepancy between them as
\begin{equation*}
\begin{split}
\Delta_A (\mathbf{c}, \mathbf{y}) = \min \{ & r \mid \exists \mathbf{z} \in \mathbb{F}_{q^m}^r, D \in \mathbb{F}_q^{N \times r} \\
 & \textrm{ with } \mathbf{y} = \mathbf{c} A^T + \mathbf{z} D^T \} = {\rm wt_R}(\mathbf{y} - \mathbf{c} A^T).
\end{split}
\end{equation*}
\end{definition}

Fix nonnegative integers $ \rho, t $, with $ {\rm Rk}(A) \geq n - \rho $. We will assume that, if $ \mathbf{c} \in \mathbb{F}_{q^m}^n $ is sent and $ \mathbf{y} \in \mathbb{F}_{q^m}^N $ is received, then $ \Delta_A(\mathbf{c},\mathbf{y}) \leq t $, or equivalently, that $ \mathbf{y} = \mathbf{c} A^T + \mathbf{e} $, with $ {\rm wt_R}(\mathbf{e}) \leq t $. Define $ L = {\rm row}(A) $. We will denote $ \widetilde{A} \subset A $ if $ \widetilde{A} $ is a submatrix of $ A $ that is a generator matrix of $ L $. 

Next we recall the decoder in \cite{on-metrics} and present a slightly different one.

\begin{definition}[\textbf{\cite[Equation (10)]{on-metrics}}] \label{decoder silva}
We define the decoder 
$$ \overline{\mathbf{c}} = {\rm argmin}_{\mathbf{c} \in C} \Delta_{A}(\mathbf{c}, \mathbf{y}). $$
\end{definition}

\begin{definition} \label{decoder ours}
For each $ \widetilde{A} \subset A $, we define the decoder:
$$ \widehat{\mathbf{c}} = {\rm argmin}_{\mathbf{c} \in C} \Delta_{\widetilde{A}}(\mathbf{c}, \widetilde{\mathbf{y}}), $$
where $ \widetilde{\mathbf{y}} $ is the vector obtained from $ \mathbf{y} $ taking the coordinates in the same positions as the rows of $ \widetilde{A} $. 
\end{definition}

We will say that one of the previous decoders is infallible \cite[Section III.A]{on-metrics} if $ \widehat{\mathbf{c}} = \mathbf{c} $ (or $ \overline{\mathbf{c}} = \mathbf{c} $), when $ \mathbf{c} $ is the sent message, for every $ \mathbf{c} \in C $.

In \cite{rgrw, on-metrics}, sufficient and necessary conditions for the decoder corresponding to $ A $ being infallible are given. We will now state that the same conditions are valid for the decoders corresponding to all the submatrices $ \widetilde{A} $. In particular, all of them give the correct (and thus, the same) answer.
 
The main difference is that now the proof only relies on Proposition \ref{info spaces distance 1} and Proposition \ref{info spaces distance 2}, where we do not need the machinery developed in \cite{rgrw, on-metrics}, in total analogy with the Hamming case, as proven in \cite[Theorem 1.5.1]{pless}, and for the decoding, we do not need all rows in $ A $. Moreover, although it is not difficult to adapt the proof in \cite[Theorem 4]{rgrw} for $ \mathbb{F}_q $-linear coding schemes, our proof works for any (arbitrary) scheme.
 
\begin{theorem}
Given an (arbitrary) code $ C \subset \mathbb{F}_{q^m}^n $, if $ d_R(C) > 2t + \rho $, then the decoders in Definition \ref{decoder ours} are infallible for every $ \widetilde{A} \subset A $, and in particular, they all give the same answer. If  $ d_R(C) \leq 2t + \rho $, then there exists a matrix $ A \in \mathbb{F}_q^{N \times n} $ such that for every $ \widetilde{A} \subset A $, the decoder in Definition \ref{decoder ours} is not infallible.
\end{theorem}
\begin{proof}
First, assume $ d_R(C) > 2t + \rho $ and fix a matrix $ A \in \mathbb{F}_q^{N \times n} $ and $ \widetilde{A} \subset A $. Assume also that the sent message is $ \mathbf{c} \in C $ and we receive $ \mathbf{y} = \mathbf{c} A^T + \mathbf{e} $, with $ {\rm wt_R}(\mathbf{e}) \leq t $. Define $ \widetilde{\mathbf{y}} $ and $ \widetilde{\mathbf{e}} $ as the vectors obtained from $ \mathbf{y} $ and $ \mathbf{e} $, respectively, taking the coordinates in the same positions as the rows in $ \widetilde{A} $. Therefore, $ \widetilde{\mathbf{y}} = \mathbf{c} \widetilde{A}^T + \widetilde{\mathbf{e}} $.

We have that $ {\rm Rk}(\widetilde{A}) = {\rm Rk}(A) $ and $ {\rm wt_R}(\widetilde{\mathbf{e}}) \leq {\rm wt_R}(\mathbf{e}) \leq t $, and on the other hand,
$$ \Delta_{\widetilde{A}}(\mathbf{c}, \widetilde{\mathbf{y}}) = {\rm wt_R}(\widetilde{\mathbf{e}}) = {\rm wt_R}(\widetilde{\mathbf{e}}A^\prime), $$
where $ A^\prime \widetilde{A}^T = I $. 

Now, $ \mathbf{c} \widetilde{A}^T A^\prime = \pi_{L,L^\prime}(\mathbf{c}) $ by Lemma \ref{lemma projection}. Since $ d_R(C^L) > 2t $ by Proposition \ref{info spaces distance 2}, and since $ L $ is an information space for $ C $ by Proposition \ref{info spaces distance 1}, $ \mathbf{c} $ is the only vector in $ C $ with $ d_R (\widetilde{\mathbf{y}} A^\prime, \pi_{L,L^\prime}(\mathbf{c}) ) \leq t $, and we are done.

Finally, if $ d_R(C) \leq 2t + \rho $, then take $ A $ such that $ \dim(L) = n - \rho $ and $ d_R(C^L) = d_R(C) - \rho \leq 2t $, which exists by Proposition \ref{info spaces distance 2}. Then, take $ \widetilde{A} \subset A $ and $ \mathbf{c}, \mathbf{c}^\prime \in C $ such that $ d_R(\pi_{L,L^\prime}(\mathbf{c}),\pi_{L,L^\prime}(\mathbf{c}^\prime)) = d_R(\mathbf{c} \widetilde{A}^T, \mathbf{c}^\prime \widetilde{A}^T) \leq 2t $. There exists $ \mathbf{e}, \mathbf{e}^\prime \in \mathbb{F}_{q^m}^N $ such that $ {\rm wt_R}(\mathbf{e}), {\rm wt_R}(\mathbf{e}^\prime) \leq t $ and $ \mathbf{c} \widetilde{A}^T + \widetilde{\mathbf{e}} = \mathbf{c}^\prime \widetilde{A}^T + \widetilde{\mathbf{e}}^\prime $, and hence the decoder associated with $ \widetilde{A} $ gives both $ \mathbf{c} $ and $ \mathbf{c}^\prime $ as solutions.
\end{proof}

To adapt this to (arbitrary) coding schemes, we just need to replace distances between vectors by distances between cosets
$$ d_R(C_{\mathbf{x}},C_{\mathbf{x}^\prime}) = \min \{ d_R(\mathbf{c},\mathbf{c}^\prime) \mid \mathbf{c} \in C_{\mathbf{x}}, \mathbf{c}^\prime \in C_{\mathbf{x}^\prime} \}, $$
and the choice of vectors in $ C $ by the choice of representatives of a coset $ C_{\mathbf{x}} $ in $ \mathcal{P}_{\mathcal{S}} $. \\

\appendix

\section{The role of $ C_1^L / C_2^L $ in information leakage} \label{app 1}

In this appendix we explain the role of $ C_1^L / C_2^L $ in information leakage beyond the expression (\ref{info leakage punctured 2}). Let the notation be as in Subsection \ref{erasure sect}.

If the adversary knows the matrix $ B $, then he or she may obtain $ \pi_{L,L^\prime}(\mathbf{c}) = \mathbf{c} \widetilde{B}^T \widetilde{B}^\prime $, where $ \widetilde{B} $ is a submatrix of $ B $ that is a generator matrix of $ L $, and $ \widetilde{B}^\prime \widetilde{B}^T = I $. Assuming uniform distributions, it can be shown that the adversary still obtains the same amount of information from $ \pi_{L,L^\prime}(\mathbf{c}) $:
\begin{equation}
{\rm I}(S; X B^T) = {\rm I}(S; \pi_{L,L^\prime}(X)) = \dim(C_1^L / C_2^L).
\label{info leakage}
\end{equation}

Actually, we can effectively compute the set of possible sent messages, regardless of the distributions used. If $ \psi : \mathbb{F}_{q^m}^\ell \longrightarrow W $ is the map in Definition \ref{definition NLCP}, we can see both $ \psi $ and $ \pi_{L,L^\prime} $ as maps
$$ \mathbb{F}_{q^m}^\ell \stackrel{\psi}{\longrightarrow} C_1 / C_2 \stackrel{\pi_{L,L^\prime}}{\longrightarrow} C_1^L / C_2^L, $$
where $ \psi $ is an isomorphism and $ \pi_{L,L^\prime} $ is surjective. Therefore, knowing $ \mathbf{c}^\prime = \pi_{L,L^\prime}(\mathbf{c} + C_2) = \pi_{L,L^\prime}(\psi(\mathbf{x})) $, where $ \mathbf{c} = \psi(\mathbf{x}) $, we can obtain the set of possible sent messages, which is 
$$ (\pi_{L,L^\prime} \circ \psi)^{-1}(\mathbf{c}^\prime) = \mathbf{x} + \ker(\pi_{L,L^\prime} \circ \psi), $$
regardless of the distribution, and in the case of uniform distributions, $ \dim (\ker(\pi_{L,L^\prime} \circ \psi)) = \ell - \dim(C_1^L / C_2^L) = H(S) - I(S;\pi_{L,L^\prime}(X)) = H(S | \pi_{L,L^\prime}(X)) $.

Moreover, if we know $ B $, we can obtain all vectors in $ \mathbf{x} + \ker(\pi_{L,L^\prime} \circ \psi) $ by performing matrix multiplications and solving systems of linear equations. 

Assume that $ G_1, G_2, G^\prime $ are generator matrices of $ C_1, C_2, W $, respectively, where $ C_1 = C_2 \oplus W $, and the first rows of $ G_1 $ are the rows in $ G_2 $, and the last rows are the rows in $ G^\prime $. Then, for a message $ \mathbf{x} \in \mathbb{F}_{q^m}^\ell $, the encoding consists in generating uniformly at random a vector $ \mathbf{x}_2 \in \mathbb{F}_{q^m}^{k_2} $ and defining $ \mathbf{c} = \mathbf{x}_2 G_2 + \mathbf{x} G^\prime = (\mathbf{x}_2, \mathbf{x}) G_1 $. Therefore, the projections onto the last $ \ell $ coordinates of the solutions of the system $ \pi_{L,L^\prime}(\mathbf{c}) = \widetilde{\mathbf{x}} (G_1 \widetilde{B}^T \widetilde{B}^\prime ) $ will be all the vectors in $ \mathbf{x} + \ker(\pi_{L,L^\prime} \circ \psi) $. 

If $ L $ is an information space for $ C_2 \varsubsetneq C_1 $, i.e., $ \dim(C_1^L/C_2^L) = \ell $, then all solutions of the previous system coincide in the last $ \ell $ coordinates, which constitute the original message $ \mathbf{x} \in \mathbb{F}_{q^m}^\ell $.

\section{Alternative proof of the duality theorem} \label{app 2}

We will now give a different proof of the duality Theorem \ref{wei duality} (proven in \cite{jerome}) that follows from Proposition \ref{pre duality}. Note that a theorem analogous to Wei's duality theorem \cite[Theorem 3]{wei} has not been given for relative generalized Hamming weights, nor for the rank case. However, Proposition \ref{pre duality} and its Hamming version work for any nested linear code pair.

We will need the following lemma:

\begin{lemma}[\textbf{\cite[Lemma 4]{rgrw}}]
For any linear code $ C \subset \mathbb{F}_{q^m}^n $ and any $ 1 \leq r \leq k $, we have that
$$ d_{R,r}(C) = \min \{ j \mid \max \{ \dim(C_L) \mid \dim(L) = j \} = r \}. $$
\end{lemma}

\begin{proof}[Proof of Theorem \ref{wei duality}]
By monotonicity and cardinality, it is enough to prove that both sets on the right-hand side are disjoint. Assume that they are not disjoint, then there exist $ i,j,s $ such that $ d_i = j $ and $ d_s^\perp = n+1 - j $. By the previous lemma, the first equality implies that 
$$ \max \{ \dim(C_L) \mid \dim(L) = j \} = i. $$
Now take $ C_1 = C $ and $ C_2 = 0 $ in Proposition \ref{pre duality}. From the fact that $ \mathcal{B}(\mathbb{F}_{q^m}^n,C^\perp)_r = \mathcal{B}(C,0)_{\ell-r}^\perp $ and the previous lemma, the second equality implies that
$$ \max \{ \dim(C_L) \mid \dim(L) = j-1 \} = s+k-n-1+j. $$
Again by the previous lemma, $ i > s+k-n-1+j $. Now interchanging the role of $ C $ and $ C^\perp $, which also interchanges the roles of $ i,s $; the roles of $ j,n+1-j $; and the roles of $ k,n-k $; we have that $ i \leq s+k-n-1+j $, which is absurd.
\end{proof}

\section{Seeing errors as erasures} \label{app 3}

We will show now that erasure correction is equivalent to error correction if the rank support of the error vector is known. This is analogous to the fact that usual erasure correction is equivalent to usual error correction where the positions of the errors (the Hamming support of the error vector) are known. This is a basic fact used in many decoding algorithms for the Hamming distance, which now we hope can be translated to the rank case.

\begin{proposition}
Assume that $ \mathbf{c} \in C $ and $ \mathbf{y} = \mathbf{c} + \mathbf{e} $, where $ {\rm wt_R}(\mathbf{e}) = t < d_R(C) $ and $ L = G(\mathbf{e}) $. Then, $ \mathbf{c} $ is the only vector $ \mathbf{c}^\prime \in C $ such that $ {\rm wt_R}(\mathbf{y} - \mathbf{c}^\prime) < d_R(C) $ and $ L = G(\mathbf{y} - \mathbf{c}^\prime) $. 

Moreover, if $ A $ is a generator matrix of $ L^\perp $, then $ \mathbf{c} $ is the unique solution in $ C $ of the system of equations $ \mathbf{y}A^T = \mathbf{x}A^T $, where $ \mathbf{x} $ is the unknown vector.
\end{proposition}
\begin{proof}
Assume that $ \mathbf{y} = \mathbf{c} + \mathbf{e} = \mathbf{c}^\prime + \mathbf{e}^\prime $, where $ \mathbf{c}^\prime \in C $ and $ G(\mathbf{e}) = G(\mathbf{e}^\prime) $. Then $ \mathbf{y}A^T = \mathbf{c}A^T = \mathbf{c}^\prime A^T $. Since $ {\rm Rk}(A) = n - t $ and $ t < d_R(C) $, it follows from the previous theorem that $ \mathbf{c} = \mathbf{c}^\prime $. 
\end{proof}

\section*{Acknowledgement}

The author wishes to thank Ryutaroh Matsumoto, Relinde Jurrius and Ruud Pellikaan for important comments on their work, during the stay of the last author at Aalborg University, and Olav Geil and Diego Ruano for fruitful discussions and careful reading of the manuscript, including the fact that the duality theorem follows from Proposition \ref{pre duality}. The author also gratefully acknowledges the support from The Danish Council for Independent Research (Grant No. DFF-4002-00367).


\def\cprime{$'$}

\end{document}